\documentclass{article}
\usepackage[utf8]{inputenc}
\usepackage{amssymb,amsmath,amsthm}
\usepackage{graphicx}
\usepackage{proof}
\usepackage{cite}
\usepackage[UKenglish]{babel}
\usepackage[all]{xy} %xymatrix
\usepackage{hyperref}
\usepackage{times}

\hypersetup{
   colorlinks,%
   citecolor=blue,%
   filecolor=black,%
   linkcolor=red,%
   urlcolor=black
 }%超链接的颜色
 
 \newcommand{\bis}{\mathrel{\mathchoice%
{\raisebox{.3ex}{$\,
  \underline{\makebox[.7em]{$\leftrightarrow$}}\,$}}%
{\raisebox{.3ex}{$\,
  \underline{\makebox[.7em]{$\leftrightarrow$}}\,$}}%
{\raisebox{.2ex}{$\,
  \underline{\makebox[.5em]{\scriptsize$\leftrightarrow$}}\,$}}%
{\raisebox{.2ex}{$\,
  \underline{\makebox[.5em]{\scriptsize$\leftrightarrow$}}\,$}}}}

\newcommand{\lr}[1]{\langle #1 \rangle}
\newcommand{\lra}{\leftrightarrow}

\newcommand{\BP}{\textbf{P}}
\newcommand{\M}{\ensuremath{\mathcal{M}}}

\newcommand{\F}{\ensuremath{\mathcal{F}}}

\renewcommand{\phi}{\varphi}

\newcommand{\weg}[1]{}

\theoremstyle{definition}
\newtheorem{theorem}{Theorem}
\newtheorem{lemma}[theorem]{Lemma}
\newtheorem{definition}[theorem]{Definition}

\newtheorem{remark}[theorem]{Remark}
\newtheorem{proposition}[theorem]{Proposition}

\newtheorem{corollary}[theorem]{Corollary}

\title{Fitchean Ignorance and First-order Ignorance:\\ A Neighborhood Look}
\author{Jie Fan\\
\small Institute of Philosophy, Chinese Academy of Sciences;\\
\small School of Humanities, University of Chinese Academy of Sciences  \\
\small \texttt{jiefan@ucas.ac.cn}}
\date{}
\begin{document}

\maketitle

\begin{abstract}
   In a seminal work~\cite{Fine:2018}, Fine classifies several forms of ignorance, among which are Fitchean ignorance, first-order ignorance, Rumsfeld ignorance, and second-order ignorance. It is shown that there is interesting relationship among some of them, which includes that in ${\bf S4}$, all higher-order ignorance are reduced to second-order ignorance. This is thought of as a bad consequence by some researchers. It is then natural to ask how to avoid this consequence. We deal with this issue in a much more general framework. In detail, we treat the forms of Fitchean ignorance and first-order ignorance as primitive modalities and study them as first-class citizens under neighborhood semantics, in which Rumsfeld ignorance and second-order ignorance are definable. The main contributions include model-theoretical results such as expressivity and frame definability, and axiomatizations. Last but not least, by updating the neighborhood models via the intersection semantics, we extend the results to the dynamic case of public announcements, which gives us some applications to successful formulas.
   %In this paper, we investigate the bimodal logic of `ignorance of' and `ignorance whether' under the neighborhood semantics. We compare the relative expressivity of this logic and other related logics such as its two sublogics and standard epistemic logic. We investigate the frame definability of the bimodal logic. We axiomatize the bimodal logic over various classes of neighborhood frames, which includes the classical logic, the monotone logic, and the regular logic. By updating the neighborhood models via the intersection semantics, we find suitable reduction axioms and thus reduce the public announcement operators to the bimodal logic, which gives us some applications to successful formulas.
\end{abstract}


\section{Introduction}

Ignorance has been a hotly discussed theme in epistemology and many other fields since Socrates, who professed ignorance in e.g. the {\em Apology}~\cite{Bett:2010}. Just as there has been no consensus on the definition of knowledge, there has been no consensus on the definition of ignorance. Instead, there has been at least three views in the literature: the standard view, the new view, and the logical view.\footnote{The terminology ‘the standard view’ is introduced in~\cite{LeMorvan:2011}, ‘the new view’ is from~\cite{Peels:2011}, whereas the term
‘the logical view’ comes from~\cite{Fan:2016}.} The standard view thinks that ignorance is merely the negation of propositional knowledge, the new view thinks that ignorance is the lack of true belief,\footnote{For the discussion on the standard and new view, see~\cite{le2016nature} and references therein.} whereas the logical view thinks that ignorance means neither knowing nor knowing not~\cite{wiebeetal:2003,hoeketal:2004,steinsvold:2008,Fanetal:2014,Fanetal:2015,olsson2015explicating}.\footnote{To the best of our knowledge, the first to evidently  investigate ignorance from the logical view is~\cite{wiebeetal:2003} --- also see its extended journal version~\cite{hoeketal:2004}, which though includes an {\em unsound} transitive axiomatization, as shown in~\cite[pp.~102--103]{Fanetal:2015}.}

%Ignorance has been a hotly discussed theme in epistemology since Socrates, who professed ignorance in e.g. the {\em Apology}~\cite{Bett:2010}. Just as there has been no consensus on the definition of knowledge, there has been no consensus on the definition of ignorance. Instead, there has been at least three views proposed: the standard view, the new view, and the logical view.\footnote{The terminology ‘the standard view’ is introduced in~\cite{LeMorvan:2011}, ‘the new view’ is from~\cite{Peels:2011}, whereas the term ‘the logical view’ comes from~\cite{Fan:2016}.} The standard view thinks that ignorance is merely the negation of propositional knowledge~\cite{Unger:1975,Driver:1989,Goldmanetal:2009,LeMorvan:2010,LeMorvan:2011,LeMorvan:2012,LeMorvan:2013}, the new view thinks that ignorance is the lack of true belief~\cite{Peels:2011,Peels:2012,Kyle:2015}, whereas the logical view thinks that ignorance means neither knowing nor knowing not~\cite{wiebeetal:2003,hoeketal:2004,steinsvold:2008,Fanetal:2014,Fanetal:2015,olsson2015explicating}.\footnote{To the best of our knowledge, the first to evidently  investigate ignorance from the logical view is~\cite{wiebeetal:2003} --- also see its extended journal version~\cite{hoeketal:2004}, which though includes an {\em unsound} transitive axiomatization, as shown in~\cite[pp.~102--103]{Fanetal:2015}.}

%The interest in ignorance has recently begun to revive. 
Recently there has been a flurry of research on ignorance. Various forms of ignorance are proposed in the literature, such as pluralistic ignorance~\cite{o1976pluralistic,bjerring2014rationality,proietti2014ddl}, circumscriptive ignorance~\cite{konolige1982circumscriptive}, chronological ignorance~\cite{shoham1986chronological}, factive ignorance~\cite{kubyshkina2019logic}, relative ignorance~\cite{Goranko:2021}, disjunctive ignorance~\cite{Fan:2021disjunctive}. In a seminal paper~\cite{Fine:2018}, instead of discussing the definition of ignorance, Fine classifies several forms of ignorance, among which are `ignorance of (the fact that)' (also called `Fitchean ignorance' there), `first-order ignorance (whether)', `Rumsfeld ignorance' and `second-order ignorance'. One is {\em ignorant of} (the fact that) $\phi$, if $\phi$ is the case but one does not know it. One is {\em (first-order) ignorant whether} $\phi$, if one neither knows $\phi$ nor knows its negation. One is {\em Rumsfeld ignorant of} $\phi$, if one is ignorant of the fact that one is ignorant whether $\phi$. One is {\em second-order ignorant whether} $\phi$, if one is ignorant whether one is ignorant whether $\phi$.

As Fine~\cite{Fine:2018} shows, there is interesting relationship among some of the forms. For instance, within the context of the system ${\bf S4}$, second-order ignorance implies first-order ignorance; second-order ignorance implies Rumsfeld ignorance, and vice versa; one does not know one is Rumsfeld ignorant; one does not know one is second-order ignorant. However, all these results are based on the context of ${\bf S4}$.\weg{But the transitivity of knowledge is contentious, see e.g.~\cite{Williamson:2000}.} It is then natural to ask what relationship among these forms there is in other contexts, based on the following reasons: firstly, although knowledge is usually based on ${\bf S4}$ (for instance in~\cite{hintikka:1962}), ignorance is not --- it is argued on the new view that ignorance is {\em not} not-knowing~(e.g.~\cite{Peels:2011}); secondly, in the first explicitly logical studies  on ignorance~\cite{wiebeetal:2003,hoeketal:2004}, the semantic condition is arbitrary, without any restriction; moreover, in ${\bf S4}$, all higher-order ignorance are reduced to second-order ignorance --- this is called the {\em black hole} of ignorance in~\cite{Fine:2018} and a {\em quite problematic phenomenon} in~\cite[p.~1060]{bonzio2022logical}.

One may easily check that the latter two forms are definable with the former two ones. It is the former two forms that are our focus here.\footnote{Fitchean ignorance and first-order ignorance correspond to important metaphysical concepts --- accident (or `accidental truths') and contingency, respectively. For the history of the bimodal logic of contingency and accident and the importance of the two metaphysical concepts, we refer to~\cite{Fan:2019} and the reference therein.} It is important to distinguish these two forms. For instance, the Fitchean ignorance satisfies the so-called {\em Factivity Principle} (that is, if an agent is ignorant of $\phi$ then $\phi$ is true), but the first-order ignorance does not.\footnote{In a recent work~\cite[p.~7]{bonzio2022logical}, the authors seem to think that ignorance has only one form, and say that ignorance should satisfy Factivity Principle since knowledge does.} Moreover, since the operators of the two forms and their duals are not normal, the logic of Fitchean ignorance and first-order ignorance is not normal. As is well known, neighborhood semantics has been a standard semantics tool for non-normal modal logics since its introduction in 1970~\cite{Montague:1970,Scott:1970,Chellas1980,pacuit2017neighborhood}. In the current paper, we will investigate the logical properties of the two forms of ignorance and their relationship under the neighborhood semantics. As we will show, there is interesting relationship among first-order ignorance, second-order ignorance, and Rumsfeld ignorance. For example, under any condition, Rumsfeld ignorance implies first-order ignorance, and second-order ignorance plus first-order ignorance implies Rumsfeld ignorance, whereas under the condition $(c)$, Rumsfeld ignorance implies second-order ignorance, and thus Rumsfeld ignorance amounts to second-order ignorance plus first-order ignorance. However, similar to the case for relational semantics~\cite{Fan:2019}, the situation may become quite involved if we study the two notions in a unified framework under the neighborhood semantics. For instance, we will be confronted with a difficulty in axiomatizing the bimodal logic, since we have only one neighborhood function to deal with two modal operators uniformly, which makes it hard to find suitable
interaction axioms.

The remainder of the paper is organized as follows. After briefly reviewing the syntax and the neighborhood semantics of the bimodal logic of Fitchean ignorance and first-order ignorance and also some related logics (Sec.~\ref{sec.syntaxandsemantics}), we compare the relative expressivity (Sec.~\ref{sec.expressivity}) and investigate the frame definability of the bimodal logic (Sec.~\ref{sec.framedefinability}). We axiomatize the bimodal logic over various classes of neighborhood frames (Sec.~\ref{sec.axiomatizations}). By updating the neighborhood models via the intersection semantics, we find suitable reduction axioms and thus reduce the public announcements operators to the bimodal logic, which gives us good applications to successful formulas (Sec.~\ref{sec.updating}), where, as we shall show, any combination of $p$, $\neg p$, ${\neg\bullet} p$, and ${\neg\nabla} p$ via conjunction (or, via disjunction) is successful under the intersection semantics. Finally, we conclude with some future work in Sec.~\ref{sec.conclusion}.

\section{Syntax and Neighborhood Semantics}\label{sec.syntaxandsemantics}

This section introduces the languages and their neighborhood semantics involved in this paper.

Fix a nonempty set $\BP$ of propositional variables, and let $p\in\BP$. In what follows, $\mathcal{L}(\Box)$ is the language of standard epistemic logic, $\mathcal{L}(\nabla)$ is the language of the logic of (first-order) ignorance, $\mathcal{L}(\bullet)$ is the language of the logic of Fitchean ignorance\footnote{$\mathcal{L}(\bullet)$ is also called `the logic of essence and accident' or `the logic of unknown truths', see e.g.~\cite{Marcos:2005,steinsvold:2008}.}, and $\mathcal{L}(\nabla,\bullet)$ is the language of the bimodal logic of Fitchean ignorance and first-order ignorance. We will mainly focus on $\mathcal{L}(\nabla,\bullet)$. For the sake of simplicity, we only exhibit the single-agent languages, but all our results also apply to
multi-agent cases.

%Fix a nonempty set $\BP$ of propositional variables, and let $p\in\BP$. In what follows, $\mathcal{L}(\Diamond)$ is the language of standard modal logic, $\mathcal{L}(\nabla)$ is the language of contingency logic, $\mathcal{L}(\bullet)$ is the language of accident logic, and $\mathcal{L}(\nabla,\bullet)$ is the language of the logic of contingency and accident. We will mainly focus on $\mathcal{L}(\nabla,\bullet)$.
\begin{definition}[Languages]  
\[
\begin{array}{llll}
   \mathcal{L}(\Box)&\phi & ::= & p\mid \neg\phi\mid \phi\land\phi\mid \Box\phi\\
   \mathcal{L}(\nabla)&\phi & ::= & p\mid \neg\phi\mid \phi\land\phi\mid \nabla\phi\\
   \mathcal{L}(\bullet)&\phi & ::= & p\mid \neg\phi\mid \phi\land\phi\mid \bullet\phi\\
   \mathcal{L}(\nabla,\bullet)&\phi & ::= & p\mid \neg\phi\mid \phi\land\phi\mid \nabla\phi\mid \bullet\phi\\
\end{array}
\]
\end{definition}

$\Box\phi$ is read ``one knows that $\phi$'', $\nabla\phi$ is read ``one is {\em (first-order) ignorant whether} $\phi$'', and $\bullet\phi$ is read ``one is {\em ignorant of} (the fact that) $\phi$'', or ``$\phi$ is an unknown truth''. In the metaphysical setting, $\nabla\phi$ and $\bullet\phi$ are read, respectively, ``it is contingent that $\phi$'' and ``it is accidental that $\phi$''. Among other connectives, $\Diamond\phi$, $\Delta\phi$, and $\circ\phi$ abbreviate, respectively, $\neg\Box\neg\phi$, ${\neg\nabla}\phi$, and ${\neg\bullet}\phi$, read ``it is epistemically possible that $\phi$'', `` one knows whether $\phi$'', and ``one is non-ignorant of $\phi$''.

Note that the forms of `Rumsfeld ignorance (of $\phi$)' and `second-ignorance (whether $\phi$)' can be defined as, respectively, $\bullet\nabla\phi$ and $\nabla\nabla\phi$.

%$\Diamond\phi$ is read ``it is possible that $\phi$'', $\nabla\phi$ is read ``it is contingent that $\phi$'', and $\bullet\phi$ is read ``it is accident that $\phi$''. In an epistemic setting, $\nabla\phi$ and $\bullet\phi$ are read, respectively, ``one is ignorant about $\phi$'' and ``$\phi$ is an unknown truth''.\footnote{In Fine's term in~\cite{Fine:2018}, $\nabla\phi$ and $\bullet\phi$ are read ``one is (first-order) ignorant whether $\phi$'' and ``one is ignorant of the fact that $\phi$'', respectively.} Among other connectives, $\Box\phi$, $\Delta\phi$, and $\circ\phi$ abbreviate, respectively, $\neg\Diamond\neg\phi$, ${\neg\nabla}\phi$, and ${\neg\bullet}\phi$, read ``it is necessary that $\phi$'', ``it is non-contingent that $\phi$'', and ``it is essential that $\phi$''.

The above languages are interpreted over neighborhood models.
\begin{definition}[Neighborhood structures]
A {\em (neighborhood) model} is a triple $\M=\lr{S,N,V}$, where $S$ is a nonempty set of states (also called `points' or `possible worlds', $N$ is a neighborhood function from $S$ to $\mathcal{P}(\mathcal{P}(S))$, and $V$ is a valuation function. A {\em (neighborhood) frame} is a model without a valuation; in this case, we say that the model is based on the frame. A {\em pointed model} is a pair of a model with a point in it. Given an $s\in S$, an element of $N(s)$ is called `a neighborhood of $s$'.
\end{definition}

The following list of neighborhood properties come from~\cite[Def.~3]{FanvD:neighborhood}.
\begin{definition}[Neighborhood properties]\label{def.properties} Let $\mathcal{F}=\lr{S,N}$ be a frame, and $\M$ be a model based on $\mathcal{F}$. Let $s\in S$ and $X,Y\subseteq S$. We define various neighborhood properties as follows.
\begin{itemize}
\item $(n)$: $N(s)$ \emph{contains the unit}, if $S\in N(s)$. % Sometimes we say $N(s)$\emph{ satisfies condition (n)}.

\item $(r)$: $N(s)$ \emph{contains its core}, if $\bigcap N(s)\in N(s)$\weg{, where $\bigcap N(s)=\bigcap\{X\mid X\in N(s)\}$}. % Sometimes we say $N(s)$ \emph{satisfies (r)}.

\item $(i)$: $N(s)$ \emph{is closed under intersections}, if $X,Y\in N(s)$ implies $X\cap Y\in N(s)$. % Sometimes we say $N(s)$\emph{ satisfies condition (i)}.

\item $(s)$: $N(s)$ is \emph{supplemented}, or \emph{closed under supersets}, if $X\in N(s)$ and $X\subseteq Y\subseteq S$ implies $Y\in N(s)$. % Sometimes we say $N(s)$\emph{ satisfies condition (s)}.

\item $(c)$: $N(s)$ is\emph{ closed under complements}, if $X\in N(s)$ implies $S\backslash X\in N(s)$.\footnote{The property $(c)$ provides a new perspective for $\mathcal{L}(\nabla)$, see~\cite{Fan:2018} for details.}
% Sometimes we say $N(s)$\emph{ satisfies condition (c)}.

\item $(d)$: $X\in N(s)$ implies $S\backslash X\notin N(s)$.

\item $(t)$:  $X\in N(s)$ implies $s\in X$.

\item  $(b)$: $s\in X$ implies $\{u\in S\mid S\backslash X\notin N(u)\}\in N(s)$.

\item $(4)$:  $X\in N(s)$ implies $\{u\in S \mid X\in N(u)\}\in N(s)$.

\item  $(5)$:  $X\notin N(s)$ implies $\{u\in S \mid X\notin N(u)\}\in N(s)$.
\end{itemize}

The function $N$ possesses such a property, if for all $s\in S$, $N(s)$ has the property. $\mathcal{F}$ (and $\M$) has a property, if $N$ has. In particular, we say that $\mathcal{F}$ (and $\M$) is {\em monotone}, if $N$ has $(s)$. $\mathcal{F}$ (and $\M$) is a {\em quasi-filter}, if $N$ has $(i)$ and $(s)$; $\mathcal{F}$ (and $\M$) is a {\em filter}, if $N$ has also $(n)$.

Also, in what follows, we will use $\mathbb{C}_n$ to denote the class of $(n)$-models, and similarly for $\mathbb{C}_r$, etc. We use $\mathbb{C}_{\text{all}}$ for the class of all neighborhood models.
\end{definition}

\begin{definition}[Semantics] Let $\M=\lr{S,N,V}$ be a model. Given a pointed model $(\M,s)$, the truth condition of formulas is defined recursively as follows:
\[\begin{array}{|lll|}
    \hline
    \M,s\vDash p & \iff & s\in V(p)\\
    \M,s\vDash\neg\phi & \iff & \M,s\nvDash\phi\\
    \M,s\vDash\phi\land\psi&\iff &\M,s\vDash\phi\text{ and }\M,s\vDash\psi\\
    \M,s\vDash\Box\phi & \iff  & \phi^\M\in N(s) \\
    %\M,s\vDash\Diamond\phi & \iff & S\backslash\phi^\M\notin N(s)\\
    \M,s\vDash\nabla\phi &\iff & \phi^\M\notin N(s)\text{ and }S\backslash\phi^\M\notin N(s) \\
    \M,s\vDash\bullet\phi & \iff & \M,s\vDash\phi\text{ and }\phi^\M\notin N(s)\\
    \hline
\end{array}\]
where $\phi^\M$ denotes the {\em truth set} of $\phi$ in $\M$, in symbols, $\phi^\M=\{s\in S\mid \M,s\vDash\phi\}$; given a set $X\subseteq S$, $S\backslash X$ denotes the complement of $X$ with respect to $S$.

We say that $\phi$ is {\em true} in $(\M,s)$, if $\M,s\vDash\phi$; we say that $\phi$ is valid on a model $\M$, notation: $\M\vDash\phi$, if for all $s$ in $\M$, we have $\M,s\vDash\phi$; we say that $\phi$ is valid on a frame $\mathcal{F}$, notation: $\mathcal{F}\vDash\phi$, if for all $\M$ based on $\mathcal{F}$, we have $\M\vDash\phi$; we say that $\phi$ is valid over a class $\mathbb{F}$ of frames, notation: $\mathbb{F}\vDash\phi$, if for all $\mathcal{F}$ in $\mathbb{F}$, we have $\mathcal{F}\vDash\phi$; we say that $\phi$ is satisfiable over the class $\mathbb{F}$, if $\mathbb{F}\nvDash\neg\phi$. Similar notions go to a set of formulas.
\end{definition}

For the sake of reference, we also list the semantics of the aforementioned defined modalities as follows:
\[
\begin{array}{lll}
    \M,s\vDash\Diamond\phi & \iff & S\backslash\phi^\M\notin N(s)\\
    %\M,s\vDash\Box\phi & \iff  & \phi^\M\in N(s) \\
    \M,s\vDash\Delta\phi & \iff & \phi^\M\in N(s)\text{ or }S\backslash\phi^\M\in N(s)\\
    \M,s\vDash\circ\phi & \iff & \M,s\vDash\phi\text{ implies }\phi^\M\in N(s).\\
     & 
\end{array}
\]

\section{Expressivity}\label{sec.expressivity}

In this section, we compare the relative expressivity of $\mathcal{L}(\nabla,\bullet)$ and other languages introduced before, over various classes of neighborhood models. Some expressivity results over the class of relational models have been obtained in~\cite{Fan:2019} and~\cite{Fan:2021}.

To make our presentation self-contained, we introduce some necessary technical terms.
\begin{definition}\label{def.expressivity} Let $\mathcal{L}_1$ and $\mathcal{L}_2$ be two languages that are interpreted on the same class of models $\mathbb{C}$, where $\mathbb{C}$ ranges over classes of models  which are models for $\mathcal{L}_1$ and for $\mathcal{L}_2$.
\begin{itemize}
\item $\mathcal{L}_2$ is {\em at least as expressive as $\mathcal{L}_1$ over $\mathbb{C}$}, notation: $\mathcal{L}_1\preceq \mathcal{L}_2[\mathbb{C}]$, if for all $\phi\in\mathcal{L}_1$, there exists $\psi\in\mathcal{L}_2$ such that for all $\M\in\mathbb{C}$ and all $s$ in $\M$, we have that $\M,s\vDash\phi$ iff $\M,s\vDash\psi$.
\item $\mathcal{L}_1$ and $\mathcal{L}_2$ are {\em equally expressive over $\mathbb{C}$}, notation: $\mathcal{L}_1\equiv \mathcal{L}_2[\mathbb{C}]$, if $\mathcal{L}_1\preceq \mathcal{L}_2[\mathbb{C}]$ and $\mathcal{L}_2\preceq \mathcal{L}_1[\mathbb{C}]$.
\item $\mathcal{L}_1$ is {\em less expressive than $\mathcal{L}_2$ over $\mathbb{C}$}, notation: $\mathcal{L}_1\prec \mathcal{L}_2[\mathbb{C}]$, if $\mathcal{L}_1\preceq \mathcal{L}_2[\mathbb{C}]$ but $\mathcal{L}_2\not\preceq \mathcal{L}_1[\mathbb{C}]$.
\item $\mathcal{L}_1$ and $\mathcal{L}_2$ are incomparable in expressivity over $\mathbb{C}$, notation: $\mathcal{L}_1\asymp\mathcal{L}_2[\mathbb{C}]$, if $\mathcal{L}_1\not\preceq\mathcal{L}_2[\mathbb{C}]$ and $\mathcal{L}_2\not\preceq\mathcal{L}_1[\mathbb{C}]$.
\end{itemize}
\end{definition}

It turns out that over the class of $(c)$-models and the class of $(t)$-models, $\mathcal{L}(\nabla)$ is at least as expressive as $\mathcal{L}(\bullet)$ (Prop.~\ref{prop.exp-nabla-bullet-c} and Prop.~\ref{prop.exp-nabla-bullet-t}), whereas $\mathcal{L}(\nabla)$ is {\em not} at least as expressive as $\mathcal{L}(\bullet)$ over the class of models possessing either of other eight neighborhood properties (Prop.~\ref{prop.exp-nabla-bullet-risd}-Prop.~\ref{prop.exp-nabla-bullet-45}).

\begin{proposition}\label{prop.exp-nabla-bullet-risd}
$\mathcal{L}(\bullet)\not\preceq\mathcal{L}(\nabla)[\mathbb{C}]$, where $\mathbb{C}\in\{\mathbb{C}_\text{all},\mathbb{C}_r,\mathbb{C}_i,\mathbb{C}_s,\mathbb{C}_d\}$.
%is either the class of all models, or the class of models satisfying $(r)$ or $(i)$ or $(s)$ or $(d)$.
%$\mathcal{L}(\nabla)$ is not at least as expressive as $\mathcal{L}(\bullet)$ over the class of all models, models satisfying $(r)$ or $(i)$ or $(s)$ or $(d)$.
\end{proposition}

\begin{proof}
Consider the following models, which comes from~\cite[Prop.~2]{FanvD:neighborhood}. An arrow from a state $x$ to a set $X$ means that $X$ is a neighborhood of $x$ (Idem for other arrows).
$$
%\hspace{-1cm}
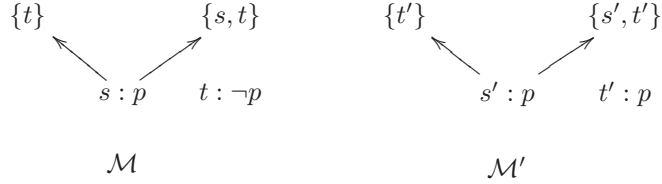
\xymatrix@C-10pt@R-10pt{\{t\}&&\{s,t\}\\
&s:p\ar[ul]\ar[ur]&t:\neg p \\
&\M&}
\qquad
\qquad
\xymatrix@L-10pt@C-10pt@R-10pt{\{t'\}&&\{s',t'\}\\
&s':p\ar[ul]\ar[ur]& t':p\\
&\M'&}
$$

It has been shown in~\cite[Prop.~2]{FanvD:neighborhood} that both $\M$ and $\M'$ satisfy $(r)$, $(i)$, $(s)$ and $(d)$, and $(\M,s)$ and $(\M',s')$ cannot be distinguished by $\mathcal{L}(\nabla)$.

However, both pointed models can be distinguished by an $\mathcal{L}(\bullet)$. To see this, note that $p^\M=\{s\}$ and $\{s\}\notin N(s)$, and thus $\M,s\vDash\bullet p$, whereas $\M',s'\nvDash \bullet p$, as $p^{\M'}=\{s',t'\}\in N'(s')$.
\end{proof}

\begin{proposition}\label{prop.exp-nabla-bullet-nb}
$\mathcal{L}(\bullet)\not\preceq\mathcal{L}(\nabla)[\mathbb{C}]$, where $\mathbb{C}\in\{\mathbb{C}_n,\mathbb{C}_b\}$.
%$\mathcal{L}(\nabla)$ is not at least as expressive as $\mathcal{L}(\bullet)$ over the class of models satisfying $(n)$ or $(b)$.
\end{proposition}

\begin{proof}
Consider the following models, which comes from~\cite[Prop.~3]{FanvD:neighborhood}:
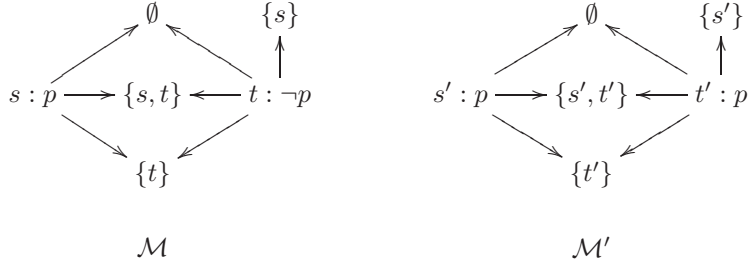
$$
%\hspace{-1cm}
\xymatrix@L-10pt@C-5pt@R-10pt{&\emptyset&\{s\}\\
s:p\ar[ur]\ar[dr]\ar[r]&\{s,t\}&t:\neg p\ar[ul]\ar[u]\ar[dl]\ar[l]\\
&\{t\}&\\
&\M&}
\qquad
\qquad
\xymatrix@L-10pt@C-5pt@R-10pt{&\emptyset&\{s'\}\\
s':p\ar[ur]\ar[r]\ar[dr]&\{s',t'\}& t':p\ar[ul]\ar[u]\ar[dl]\ar[l]\\
&\{t'\}&\\
&\M'&}
$$

It has been shown in~\cite[Prop.~3]{FanvD:neighborhood} that both $\M$ and $\M'$ satisfy $(n)$ and $(b)$, and $(\M,s)$ and $(\M',s')$ cannot be distinguished by $\mathcal{L}(\nabla)$.

However, both pointed models can be distinguished by an $\mathcal{L}(\bullet)$. To see this, note that $p^\M=\{s\}$ and $\{s\}\notin N(s)$, and thus $\M,s\vDash\bullet p$, whereas $\M',s'\nvDash \bullet p$, as $p^{\M'}=\{s',t'\}\in N'(s')$.
\end{proof}

\begin{proposition}\label{prop.exp-nabla-bullet-45}
$\mathcal{L}(\bullet)\not\preceq\mathcal{L}(\nabla)[\mathbb{C}]$, where $\mathbb{C}\in\{\mathbb{C}_4,\mathbb{C}_5\}$.
%$\mathcal{L}(\nabla)$ is not at least as expressive as $\mathcal{L}(\bullet)$ over the class of models satisfying $(4)$ or $(5)$.
\end{proposition}

\begin{proof}
Consider the following models, which is a revision of the figures in~\cite[Prop.~4]{FanvD:neighborhood}:
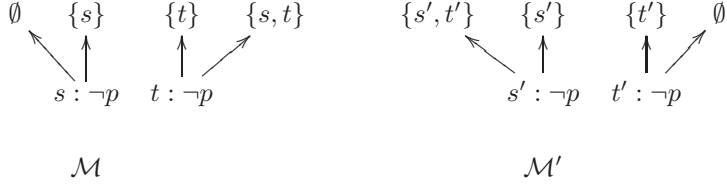
$$
\hspace{-1cm}
\xymatrix@L-10pt@C-18pt@R-10pt{\emptyset&\{s\}&\{t\}&\{s,t\}&&&\{s',t'\}&\{s'\}&\{t'\}&\emptyset\\
&s:\neg p\ar[ul]\ar[u]&t:\neg p\ar[u]\ar[ur] &&&&&s':\neg p\ar[ul]\ar[u]& t':\neg p\ar[u]\ar[ur]&\\
&\M&&&&&&\M'&&}$$

Firstly, $\M$ and $\M'$ satisfy $(4)$ and $(5)$. In what follows we only show the claim for $\M$; the proof for the case $\M'$ is analogous.
\begin{itemize}
\item[-] For $(4)$: Suppose that $X\in N(s)$. Then $X=\emptyset$ or $X=\{s\}$. Notice that $\{u\mid X\in N(u)\}=\{s\}\in N(s)$. Similarly, we can demonstrate that $(4)$ holds for $N(t)$.
\item[-] For $(5)$: Assume that $X\notin N(s)$. Then $X=\{t\}$ or $X=\{s,t\}$. Notice that $\{u\mid X\notin N(u)\}=\{s\}\in N(s)$. A similar argument goes for $N(t)$.
\end{itemize}    

Secondly, $(\M,s)$ and $(\M',s')$ cannot be distinguished by $\mathcal{L}(\nabla)$, that is to say, for all $\phi\in\mathcal{L}(\nabla)$, we have that $\M,s\vDash\phi$ iff $\M',s'\vDash\phi$. The proof goes by induction on $\phi$, where the only nontrivial case is $\nabla\phi$.
By semantics, we have the following equivalences:
\[
\begin{array}{ll}
       & \M,s\vDash\nabla\phi  \\
    \iff & \phi^\M\notin N(s)\text{ and }(\neg\phi)^\M\notin N(s) \\
    \iff & \phi^\M\notin \{\emptyset,\{s\}\}\text{ and }(\neg\phi)^\M\notin\{\emptyset,\{s\}\}\\
    \iff & \phi^\M\neq \emptyset\text{ and }\phi^\M\neq \{s\}\text{ and }(\neg\phi)^\M\neq \emptyset\text{ and }(\neg\phi)^\M\neq \{s\}\\
    \iff & \phi^\M\neq \emptyset\text{ and }\phi^\M\neq \{s\}\text{ and }\phi^\M\neq \{s,t\}\text{ and }\phi^\M\neq \{t\}\\
    \iff & \text{false}\\
\end{array}
\]
\[
\begin{array}{ll}
       & \M',s'\vDash\nabla\phi  \\
    \iff & \phi^{\M'}\notin N'(s')\text{ and }(\neg\phi)^{\M'}\notin N'(s') \\
    \iff & \phi^{\M'}\notin \{\{s',t'\},\{s'\}\}\text{ and }(\neg\phi)^{\M'}\notin\{\{s',t'\},\{s'\}\}\\
    \iff & \phi^{\M'}\neq \{s',t'\}\text{ and }\phi^{\M'}\neq \{s'\}\text{ and }(\neg\phi)^{\M'}\neq \{s',t'\}\text{ and }(\neg\phi)^{\M'}\neq \{s'\}\\
    \iff & \phi^{\M'}\neq \{s',t'\}\text{ and }\phi^{\M'}\neq \{s'\}\text{ and }\phi^{\M'}\neq \emptyset\text{ and }\phi^{\M'}\neq \{t'\}\\
    \iff & \text{false}\\
\end{array}
\]

In either case, the penultimate line of the proof merely states that $\phi$ cannot be interpreted on the related model: its denotation is {\em not} one of all possible subsets of the domain. We conclude that $\M,s\vDash\nabla\phi$ iff $\M',s'\vDash\nabla\phi$.

Finally, we show that $(\M,s)$ and $(\M',s')$ can be distinguished by $\mathcal{L}(\bullet)$. To see this, note that $(\neg p)^\M=\{s,t\}\notin N(s)$, and thus $\M,s\vDash\bullet\neg p$. However, since $(\neg p)^{\M'}=\{s',t'\}\in N'(s')$, we have $\M,s\nvDash\bullet\neg p$.
\end{proof}

\begin{remark}\label{remark}
The reader may ask whether the figure in~\cite[Prop.~4]{FanvD:neighborhood} (as below) applies to the above proposition.
$$
\hspace{-1cm}
\xymatrix@L-10pt@C-18pt@R-10pt{\emptyset&\{s,t\}&\{s\}&\{t\}&&&\{s',t'\}&\{s'\}&\{t'\}&\emptyset\\
&s:\neg p\ar[ul]\ar[u]\ar[ur]&&t:\neg p\ar[u] &&&&s':\neg p\ar[ul]\ar[u]& t':\neg p\ar[u]\ar[ur]&\\
&\M&&&&&&\M'&&}$$
The answer is negative. This is because the pointed models $(\M,s)$ and $(\M',s')$ in this figure cannot be distinguished by $\mathcal{L}(\bullet)$ either. To see this, note that $\M,s\vDash\bullet\phi$ iff $\M,s\vDash\phi$ and $\phi^\M\notin N(s)$, which by the construction of $N(s)$ implies that $s\in\phi^\M$ and $\phi^\M\neq \{s\}$ and $\phi^\M\neq \{s,t\}$, which is impossible. It then follows that $\M,s\nvDash\bullet\phi$. A similar argument can show that $\M',s'\nvDash\bullet\phi$. Therefore, $\M,s\vDash\bullet\phi$ iff $\M',s'\vDash\bullet\phi$.
\end{remark}

\begin{proposition}\label{prop.exp-nabla-bullet-c}
$\mathcal{L}(\bullet)\preceq\mathcal{L}(\nabla)[\mathbb{C}_c]$.
%$\mathcal{L}(\nabla)$ is at least as expressive as $\mathcal{L}(\bullet)$ over the class of $(c)$-models.
\end{proposition}

\begin{proof}
It suffices to show that $\bullet\phi\lra (\phi\land\nabla\phi)$ is valid over the class of $(c)$-models. Let $\M=\lr{S,N,V}$ be a $(c)$-model and $s\in S$. Suppose that $\M,s\vDash\bullet\phi$, it remains only to prove that $\M,s\vDash \phi\land\nabla\phi$. By supposition, we have $\M,s\vDash\phi$ and $\phi^{\M}\notin N(s)$. We have also $S\backslash\phi^{\M}\notin N(s)$: otherwise, by $(c)$, $S\backslash(S\backslash\phi^{\M})\in N(s)$, that is, $\phi^{\M}\in N(s)$: a contradiction. Thus $\M,s\vDash\nabla\phi$, and therefore $\M,s\vDash\phi\land\nabla\phi$. The converse is clear from the semantics.
%By~\cite[Prop.~6]{FanvD:neighborhood}, over the class of $(c)$-models, $\mathcal{L}(\nabla)$ is equally expressive as $\mathcal{L}(\Diamond)$. Moreover, by semantics, one may easily show that $\vDash\bullet\phi\lra (\phi\land\Diamond\neg\phi)$, thus $\mathcal{L}(\Diamond)$ is at least as expressive as $\mathcal{L}(\bullet)$ on any class of models, and thus also on the class of $(c)$-models. Therefore, $\mathcal{L}(\nabla)$ is at least as expressive as $\mathcal{L}(\bullet)$ over the class of $(c)$-models.
\end{proof}

\begin{proposition}\label{prop.exp-nabla-bullet-t}
$\mathcal{L}(\bullet)\preceq\mathcal{L}(\nabla)[\mathbb{C}_t]$.
%$\mathcal{L}(\nabla)$ is at least as expressive as $\mathcal{L}(\bullet)$ over the class of $(t)$-models.
\end{proposition}

\begin{proof}
It suffices to show that $\bullet\phi\lra (\phi\land\nabla\phi)$ over the class of $(t)$-models. The proof is almost the same as that in Prop.~\ref{prop.exp-nabla-bullet-c}, except that $S\backslash\phi^{\M}\notin N(s)$ (that is, $(\neg\phi)^{\M}\notin N(s)$) is obtained from $\M,s\vDash\phi$ and the property $(t)$.
%By~\cite[Prop.~5]{FanvD:neighborhood}, over the class of $(t)$-models, $\mathcal{L}(\nabla)$ is equally expressive as $\mathcal{L}(\Diamond)$. From Prop.~\ref{prop.exp-nabla-bullet-c}, we have seen that $\mathcal{L}(\Diamond)$ is at least as expressive as $\mathcal{L}(\bullet)$ on any class of models, thus also on the class of $(t)$-models. Therefore, $\mathcal{L}(\nabla)$ is at least as expressive as $\mathcal{L}(\bullet)$ over the class of $(t)$-models.
\end{proof}

Conversely, on the class of $(c)$-models and the class of $(t)$-models, $\mathcal{L}(\bullet)$ is at least as expressive as $\mathcal{L}(\nabla)$ (Prop.~\ref{prop.exp-bullet-nabla-c} and Prop.~\ref{prop.exp-bullet-nabla-t}), whereas on the class of models possessing either of other eight neighborhood properties, $\mathcal{L}(\bullet)$ is {\em not} at least as expressive as $\mathcal{L}(\nabla)$ (Prop.~\ref{prop.exp-bullet-nabla-nrisdb}-Prop.~\ref{prop.exp-bullet-nabla-5}). As a corollary, on the class of $(c)$-models and the class of $(t)$-models, $\mathcal{L}(\nabla)$, $\mathcal{L}(\bullet)$, and $\mathcal{L}(\nabla,\bullet)$ are equally expressive, whereas over the class of models possessing the eight neighborhood properties in question, $\mathcal{L}(\nabla)$ and $\mathcal{L}(\bullet)$ are both less expressive than $\mathcal{L}(\nabla,\bullet)$ (Coro.~\ref{coro.exp-nabla-bullet-nablabullet}).

\begin{proposition}\label{prop.exp-bullet-nabla-nrisdb}
%On the class of all models, the $(m)$-models, the $(c)$-models, the $(n)$-models, the $(r)$-models, $\mathcal{L}(\bullet)$ is not at least as expressive as $\mathcal{L}(\nabla)$.
$\mathcal{L}(\nabla)\not\preceq\mathcal{L}(\bullet)[\mathbb{C}]$, where $\mathbb{C}\in\{\mathbb{C}_\text{all},\mathbb{C}_n,\mathbb{C}_r,\mathbb{C}_i,\mathbb{C}_s,\mathbb{C}_d,\mathbb{C}_b\}$.
%$\mathcal{L}(\bullet)$ is not at least as expressive as $\mathcal{L}(\nabla)$ over the class of all neighborhood models, the class of neighborhood models satisfying $(n)$ or $(r)$ or $(i)$ or $(s)$ or $(d)$ or $(b)$.
\end{proposition}

\begin{proof}
Consider the following models:
$$
\xymatrix@L-5pt@C-9pt@R-5pt{&&&&&&&\{t'\}\\
\M&s:\neg p\ar[r]&\{s,t\}&t:p\ar[l]&\M'&s':\neg p\ar[r]\ar[urr]&\{s',t'\}&t':p\ar[l]}
$$
It is straightforward to check that both $\M$ and $\M'$ satisfy $(n)$, $(r)$, $(i)$, $(s)$, and $(d)$. In what follows, we show that $\M$ and $\M'$ both have the property $(b)$.
\begin{itemize}
    \item[-] For $\M$: suppose that $s\in X$. Then $X=\{s\}$ or $X=\{s,t\}$. This implies that $\{u\mid S\backslash X\notin N(u)\}=\{s,t\}\in N(s)$. Similarly, we can show that $(b)$ holds for $N(t)$.
    \item[-] For $\M'$: assume that $s'\in X$. Then $X=\{s'\}$ or $X=\{s',t'\}$. If $X=\{s'\}$, then $\{u\mid S'\backslash X\notin N'(u)\}=\{t'\}\in N'(s')$; if $X=\{s',t'\}$, then $\{u\mid S'\backslash X\notin N'(u)\}=\{s',t'\}\in N'(s')$. Now assume that $t'\in X$. Then $X=\{t'\}$ or $X=\{s',t'\}$. If $X=\{t'\}$, then $\{u\mid S'\backslash X\notin N'(u)\}=\{s',t'\}\in N'(t')$; if $X=\{s',t'\}$, we can also show that $\{u\mid S'\backslash X\notin N'(u)\}=\{s',t'\}\in N'(t')$.
\end{itemize}

Moreover, $(\M,s)$ and $(\M',s')$ cannot be distinguished by $\mathcal{L}(\bullet)$. Here we use the notion of $\bullet$-morphisms introduced in~\cite[Def.~4.1]{Fan:2020neighborhood}.\footnote{Recall that the notion of $\bullet$-morphisms is defined as follows. Let $\M=\lr{S,N,V}$ and $\M'=\lr{S',N',V'}$ be neighborhood models. A function $f:S\to S'$ is a $\bullet$-morphism from $\M$ to $\M'$, if for all $s\in S$, 
\begin{itemize}
    \item[(Var)] $s\in V(p)$ iff $f(s)\in V'(p)$ for all $p\in\BP$,
    \item[($\bullet$-Mor)] for all $X'\subseteq S'$, $[s\in f^{-1}[X']\text{ and }f^{-1}[X']\notin N(s)]\Longleftrightarrow[f(s)\in X'\text{ and }X'\notin N'(f(s))]$.
\end{itemize}
It is then demonstrated in~\cite[Prop.~4.1]{Fan:2020neighborhood} that the formulas of $\mathcal{L}(\bullet)$ are invariant under $\bullet$-morphisms. In details, let $\M$ and $\M'$ be neighborhood models, and let $f$ be a $\bullet$-morphism from $\M$ to $\M'$. Then for all $s\in S$, for all $\phi\in\mathcal{L}(\bullet)$, we have that $\M,s\vDash\phi$ iff $\M',f(s)\vDash\phi$.\label{fn.morphisms}}
Define a function $f:S\to S'$ such that $f(s)=s'$ and $f(t)=t'$. We prove that $f$ is a $\bullet$-morphism from $\M$ to $\M'$. The condition (Var) follows directly from the valuations. For the condition ($\bullet$-Mor), we first prove that it holds for $s$: assume that $s\in f^{-1}[X']$ and $f^{-1}[X']\notin N(s)$, then it must be that $X'=\{s'\}$. Then we have $f(s)=s'\in X'$ and $X'\notin N'(f(s))$. The converse is similar. In a similar way, we can show that ($\bullet$-Mor) also holds for $t$. Then by~\cite[Prop.~4.1]{Fan:2020neighborhood} (see also fn.~\ref{fn.morphisms}), we have $\M,s\vDash\phi$ iff $\M',s'\vDash\phi$ for all $\phi\in\mathcal{L}(\bullet)$.

However, these pointed models can be distinguished by $\mathcal{L}(\nabla)$. This is because $\M,s\vDash\nabla p$ (as $p^\M=\{t\}\notin N(s)$ and $(\neg p)^\M=\{s\}\notin N(s)$) and $\M',s'\nvDash\nabla p$ (as $p^{\M'}=\{t'\}\in N'(s')$).
\end{proof}

\begin{proposition}\label{prop.exp-bullet-nabla-4}
$\mathcal{L}(\nabla)\not\preceq \mathcal{L}(\bullet)[\mathbb{C}_4]$.
%$\mathcal{L}(\bullet)$ is not at least as expressive as $\mathcal{L}(\nabla)$ over the class of $(4)$-models.% or $(5)$.
\end{proposition}

\begin{proof}
Consider the following models:
$$
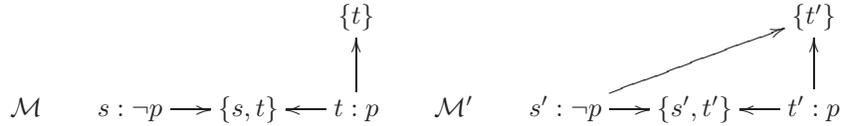
\xymatrix@L-5pt@C-9pt@R-5pt{&&&\{t\}&&&&\{t'\}\\
\M&s:\neg p\ar[r]&\{s,t\}&t:p\ar[l]\ar[u]&\M'&s':\neg p\ar[r]\ar[urr]&\{s',t'\}&t':p\ar[l]\ar[u]}
$$

Firstly, both $\M$ and $\M'$ have $(4)$.% To begin with, consider $\M$.
\begin{itemize}
    \item[-] For $\M$: Suppose that $X\in N(s)$.  Then $X=\{s,t\}$, and so $\{u\mid X\in N(u)\}=\{s,t\}\in N(s)$.
    Now assume that $X\in N(t)$. Then $X=\{t\}$ or $X=\{s,t\}$. If $X=\{t\}$, then $\{u\mid X\in N(u)\}=\{t\}\in N(t)$; if $X=\{s,t\}$, then $\{u\mid X\in N(u)\}=\{s,t\}\in N(t)$.
    %\item[-] For $(5)$: Assume that $X\notin N(s)$. Then $X=\emptyset$ or $X=\{t\}$. If $X=\emptyset$, then $\{u\mid X\notin N(u)\}=\{s,t\}\in N(s)$; if $X=\{t\}$, then $\{u\mid X\notin N(u)\}=\{s\}\in N(s)$. Similarly, we can show that $X\notin N(t)$ implies $\{u\mid X\notin N(u)\}\in N(t)$.
    \item[-] For $\M'$: Suppose that $X\in N'(s')$. Then $X=\{t'\}$ or $X=\{s',t'\}$. Either case implies that $\{u\mid X\in N'(u)\}=\{s',t'\}\in N'(s')$. Now assume that $X\in N'(t')$. Then $X=\{t'\}$ or $X=\{s',t'\}$. Again, either case implies that $\{u\mid X\in N'(u)\}=\{s',t'\}\in N'(t')$.
   % \item[-] For $(5)$: Suppose that $X\notin N(s)$. Then $X=\emptyset$, and so $\{u\mid X\notin N(u)\}=\{s,t\}\in N(s)$. Now assume that $X\notin N(t)$. Then $X=\emptyset$ or $X=\{s\}$. If $X=\emptyset$, then $\{u\mid X\notin N(u)\}=\{s,t\}\in N(t)$; if $X=\{s\}$, then $\{u\mid X\notin N(u)\}=\{t\}\in N(t)$.
\end{itemize}

Secondly, similar to the proof of the corresponding part in Prop.~\ref{prop.exp-bullet-nabla-nrisdb}, we can show that $(\M,s)$ and $(\M',s')$ cannot be distinguished by $\mathcal{L}(\bullet)$. 

It remains only to show that $(\M,s)$ and $(\M',s')$ can be distinguished by $\mathcal{L}(\nabla)$. The proof for this is analogous to that in Prop.~\ref{prop.exp-bullet-nabla-nrisdb}.
\end{proof}

\begin{proposition}\label{prop.exp-bullet-nabla-5}
$\mathcal{L}(\nabla)\not\preceq \mathcal{L}(\bullet)[\mathbb{C}_5]$.
%$\mathcal{L}(\bullet)$ is not at least as expressive as $\mathcal{L}(\nabla)$ over the class of $(5)$-models.
\end{proposition}

\begin{proof}
Consider the following models:
\[
\xymatrix@L-5pt@C-9pt@R-5pt{
     & \{s\} &&\{t\}&& && \{s'\} &&\{t'\}\\
 \M    & s:p\ar[u]\ar[urr]&\emptyset&t:p\ar[u]\ar[ull]\ar[l]\ar[r]&\{s,t\}&&\M'    & s':p\ar[u]\ar[r]\ar[urr]&\emptyset&t':p\ar[u]\ar[ull]\ar[l]\ar[r]&\{s',t'\}\\}
\]

Firstly, both $\M$ and $\M'$ possess the property $(5)$. Since for all $X\subseteq S=\{s,t\}$, $X\in N(t)$, the property $(5)$ is possessed vacuously by $N(t)$ Similarly, $(5)$ is also possessed vacuously by $N'(t')$. It remains only to show that both $N(s)$ and $N'(s')$ have $(5)$.
\begin{itemize}
    \item[-] For $N(s)$: suppose that $X\notin N(s)$, then $X=\emptyset$ or $X=\{s,t\}$. Either case implies that $\{u\in S\mid X\notin N(u)\}=\{s\}\in N(s)$.
    \item[-] For $N'(s')$: assume that $X\notin N'(s')$, then $X=\{s',t'\}$. This follows that $\{u\in S'\mid X\notin N'(u)\}=\{s'\}\in N'(s')$.
\end{itemize}

Secondly, $(\M,s)$ and $(\M',s)$ cannot be distinguished by $\mathcal{L}(\bullet)$. Again, this can be shown as the corresponding part in Prop.~\ref{prop.exp-bullet-nabla-nrisdb}.

Finally, $(\M,s)$ and $(\M',s')$ can be distinguished by $\mathcal{L}(\nabla)$. On one hand, $p^\M=\{s,t\}\notin N(s)$ and $S\backslash p^\M=\emptyset\notin N(s)$, thus $\M,s\vDash\nabla p$. On the other hand, $S'\backslash p^{\M'}=\emptyset\in N'(s')$, thus $\M',s'\nvDash\nabla p$.
\end{proof}

\begin{proposition}\label{prop.exp-bullet-nabla-c}
$\mathcal{L}(\nabla)\preceq \mathcal{L}(\bullet)[\mathbb{C}_c]$.
%$\mathcal{L}(\bullet)$ is at least as expressive as $\mathcal{L}(\nabla)$ over the class of $(c)$-models.
\end{proposition}

\begin{proof}
We claim that over the class of $(c)$-models, $\vDash\nabla\phi\lra \bullet\phi\vee{\bullet\neg}\phi$.

First, we show that $\vDash\nabla\phi\to \bullet\phi\vee{\bullet\neg}\phi$.

Let $\M=\lr{S,N,V}$ be any model and $s\in S$. Suppose that $\M,s\vDash\nabla\phi$. Then $\phi^\M\notin N(s)$ and $S\backslash\phi^\M\notin N(s)$, that is, $(\neg\phi)^\M\notin N(s)$. We have either $\M,s\vDash\phi$ or $\M,s\vDash\neg\phi$. If $\M,s\vDash\phi$, since $\phi^\M\notin N(s)$, we infer that $\M,s\vDash\bullet\phi$; if $\M,s\vDash\neg\phi$, since $(\neg\phi)^\M\notin N(s)$, we derive that $\M,s\vDash{\bullet\neg}\phi$. Therefore, $\M,s\vDash\bullet\phi\vee{\bullet\neg}\phi$. Since $(\M,s)$ is arbitrary, this establishes the validity of $\nabla\phi\to \bullet\phi\vee{\bullet\neg}\phi$.

Conversely, we prove that over the class of $(c)$-models, $\vDash \bullet\phi\vee{\bullet\neg}\phi\to \nabla\phi$.

Suppose that $\M=\lr{S,N,V}$ be a $(c)$-model and $s\in S$. Assume that $\M,s\vDash\bullet\phi\vee{\bullet\neg}\phi$. Then $\M,s\vDash \bullet\phi$ or $\M,s\vDash{\bullet\neg}\phi$. If $\M,s\vDash \bullet\phi$, then $\phi^\M\notin N(s)$. By $(c)$, we have $S\backslash \phi^\M\notin N(s)$, so $\M,s\vDash\nabla\phi$. If $\M,s\vDash{\bullet\neg}\phi$, then $(\neg\phi)^\M\notin N(s)$, namely $S\backslash\phi^\M\notin N(s)$. By $(c)$ again, we obtain $\phi^\M\notin N(s)$, and thus $\M,s\vDash\nabla\phi$. Therefore, $\M,s\vDash\nabla\phi$.
\end{proof}

\begin{proposition}\label{prop.exp-bullet-nabla-t}
$\mathcal{L}(\nabla)\preceq \mathcal{L}(\bullet)[\mathbb{C}_t]$.
%$\mathcal{L}(\bullet)$ is at least as expressive as $\mathcal{L}(\nabla)$ over the class of $(t)$-models.
%$\mathcal{L}(\nabla)$ and $\mathcal{L}(\bullet)$ are equally expressive over the class of $(t)$-models.
\end{proposition}

\begin{proof}
We claim that over the class of $(t)$-models, $\vDash\nabla\phi\lra \bullet\phi\vee\bullet\neg\phi$. The proof is almost the same as that in Prop.~\ref{prop.exp-bullet-nabla-c}, except that in the proof of the validity of $\bullet\phi\vee\bullet\neg\phi\to\nabla\phi$, $\M,s\vDash\nabla\phi$ is obtained as follows: if $\M,s\vDash\bullet\phi$, then $\M,s\vDash\phi$ and $\phi^{\M}\notin N(s)$, thus $\M,s\nvDash\neg\phi$, namely $s\notin (\neg\phi)^{\M}$, and then by $(t)$, we infer that $(\neg\phi)^{\M}\notin N(s)$, namely $S\backslash\phi^{\M}\notin N(s)$, so $\M,s\vDash\nabla\phi$; similarly, we can show that if $\M,s\vDash\bullet\neg\phi$ then $\M,s\vDash\nabla\phi$.
\end{proof}

\weg{By Props.~\ref{prop.exp-nabla-bullet-c}, \ref{prop.exp-nabla-bullet-t}, \ref{prop.exp-bullet-nabla-c}, \ref{prop.exp-bullet-nabla-t}, and the fact that $\mathcal{L}(\nabla)$ is equally expressive as $\mathcal{L}(\Box)$ over the class of $(c)$-models and over the class of $(t)$-models~\cite[Prop.~5, Prop.~6]{FanvD:neighborhood}, we have immediately the following result.

\begin{corollary}\label{coro.exp-nabla-bullet-diamond-equal}
$\mathcal{L}(\nabla)$, $\mathcal{L}(\bullet)$, and $\mathcal{L}(\Box)$ are equally expressive over the class of $(c)$-models and also over the class of $(t)$-models.
\end{corollary}}

With the above results in mind, we have the following result, which extends the expressivity results over Kripke models in~\cite{Fan:2019}.
\begin{corollary}\label{coro.exp-nabla-bullet-nablabullet}
Where $\mathbb{C}\in\{\mathbb{C}_\text{all},\mathbb{C}_n,\mathbb{C}_r,\mathbb{C}_i,\mathbb{C}_s,\mathbb{C}_d,\mathbb{C}_b,\mathbb{C}_4,\mathbb{C}_5\}$,
$\mathcal{L}(\nabla)\asymp\mathcal{L}(\bullet)[\mathbb{C}]$, and $\mathcal{L}\prec\mathbb{L}(\nabla,\bullet)[\mathbb{C}]$, where $\mathcal{L}\in\{\mathcal{L}(\nabla),\mathcal{L}(\bullet)\}$. Where $\mathbb{C}\in\{\mathbb{C}_c,\mathbb{C}_t\}$, $\mathcal{L}_1\equiv\mathcal{L}_2[\mathbb{C}]$, where $\mathcal{L}_1,\mathcal{L}_2\in\{\mathcal{L}(\nabla),\mathcal{L}(\bullet),\mathcal{L}(\nabla,\bullet)\}$.
%$\mathcal{L}(\nabla)$ and $\mathcal{L}(\bullet)$ are incomparable and both less expressive than $\mathcal{L}(\nabla,\bullet)$ over the class of all models, models satisfying $(n)$ or $(r)$ or $(i)$ or $(s)$ or $(d)$ or $(b)$ or $(4)$ or $(5)$, whereas all three logics are equally expressive over the class of $(c)$-models and over the class of $(t)$-models.
\end{corollary}

Moreover, over the class of $(c)$-models and the class of $(t)$-models, $\mathcal{L}(\nabla,\bullet)$ and $\mathcal{L}(\Box)$ are equally expressive (Prop.~\ref{prop.nablabullet-diamond-ct}), whereas over the class of models possessing either of other eight neighborhood properties except for $(d)$, $\mathcal{L}(\nabla,\bullet)$ is less expressive than $\mathcal{L}(\Box)$ (Prop.~\ref{prop.exp-nablabullet-diamond-ri45} and Prop.~\ref{prop.exp-nablabullet-diamond-nsb}).

\begin{proposition}\label{prop.exp-nablabullet-diamond-ri45}
$\mathcal{L}(\nabla,\bullet)\prec\mathcal{L}(\Box)[\mathbb{C}]$, where $\mathbb{C}\in \{\mathbb{C}_\text{all},\mathbb{C}_r,\mathbb{C}_i,\mathbb{C}_4,\mathbb{C}_5\}$.
%$\mathcal{L}(\nabla,\bullet)$ is less expressive than $\mathcal{L}(\Diamond)$ over the class of all neighborhood models, the class of neighborhood models satisfying $(r)$ or $(i)$ or $(4)$ or $(5)$.
\end{proposition}

\begin{proof}
Use Remark~\ref{remark} and~\cite[Prop.~4]{FanvD:neighborhood}.
\end{proof}

\begin{proposition}\label{prop.exp-nablabullet-diamond-nsb}
$\mathcal{L}(\nabla,\bullet)\prec\mathcal{L}(\Box)[\mathbb{C}]$, where $\mathbb{C}\in\{\mathbb{C}_n,\mathbb{C}_s,\mathbb{C}_b\}$.
%$\mathcal{L}(\nabla,\bullet)$ is less expressive than $\mathcal{L}(\Diamond)$ over the class of neighborhood models satisfying $(n)$ or $(s)$ or $(b)$.
\end{proposition}

\begin{proof}
Consider the following models $\M=\lr{S,N,V}$ and $\M'=\lr{S',N',V'}$, where $S=\{s,t\}$ and $S'=\{s',t'\}$.
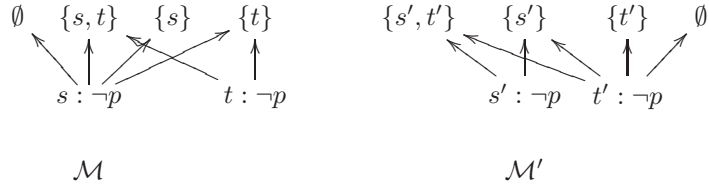
$$
\hspace{-1cm}
\xymatrix@L-10pt@C-18pt@R-10pt{\emptyset&\{s,t\}&\{s\}&\{t\}&&&\{s',t'\}&\{s'\}&\{t'\}&\emptyset\\
&s:\neg p\ar[ul]\ar[u]\ar[ur]\ar[urr]&&t:\neg p\ar[u]\ar[ull] &&&&s':\neg p\ar[ul]\ar[u]& t':\neg p\ar[u]\ar[ur]\ar[ul]\ar[ull]&\\
&\M&&&&&&\M'&&}$$

It should be straightforward to check that both $\M$ and $\M'$ possess $(n)$ and $(s)$. Moreover, both models also possess $(b)$, shown as follows. Since for all $X\subseteq S$, we have $X\in N(s)$, one may easily see that $N(s)$ has $(b)$. Similarly, we can show that $N'(t')$ has $(b)$. Besides,
\begin{itemize}
    \item[-] $N(t)$ has $(b)$: suppose that $t\in X$, then $X=\{t\}$ or $X=\{s,t\}$. The first case implies $\{u\in S\mid S\backslash X\notin N(u)\}=\{u\in S\mid \{s\}\notin N(u)\}=\{t\}\in N(t)$, and the second case implies that $\{u\in S\mid S\backslash X\notin N(u)\}=\{u\in S\mid \emptyset\notin N(u)\}=\{t\}\in N(t)$, as desired.
    \item[-] $N'(s')$ has $(b)$: assume that $s'\in X$, then $X=\{s'\}$ or $X=\{s',t'\}$. The first case entails that $\{u\in S'\mid S'\backslash X\notin N'(u)\}=\{u\in S'\mid \{t'\}\notin N'(u)\}=\{s'\}\in N'(s')$, and the second case entails that $\{u\in S'\mid S'\backslash X\notin N'(u)\}=\{u\in S'\mid \emptyset\notin N'(u)\}=\{s'\}\in N'(s')$, as desired.
\end{itemize}

Next, we show $(\M,s)$ and $(\M',s')$ cannot be distinguished by $\mathcal{L}(\nabla,\bullet)$. That is, for all $\phi\in\mathcal{L}(\nabla,\bullet)$, we have $\M,s\vDash \phi$ iff $\M',s'\vDash\phi$. The proof proceeds by induction on $\phi$, where the nontrivial cases are $\nabla \phi$ and $\bullet\phi$. The proof for the case $\bullet\phi$ is shown as in Remark~\ref{remark}. For the case $\nabla\phi$, we have the following equivalences.
\[
\begin{array}{ll}
     & \M,s\vDash\nabla\phi \\
    \iff & \phi^\M\notin N(s)\text{ and }S\backslash\phi^\M\notin N(s)\\
    \iff & \phi^\M\neq \emptyset\text{ and }\phi^\M\neq \{s,t\}\text{ and }\phi^\M\neq \{s\}\text{ and }\phi^\M\neq \{t\}\\
    \iff & \text{false}\\
\end{array}
\]
\[
\begin{array}{ll}
     & \M',s'\vDash\nabla\phi \\
    \iff & \phi^{\M'}\notin N'(s')\text{ and }S'\backslash\phi^{\M'}\notin N'(s')\\
    \iff & \phi^{\M'}\notin \{\{s',t'\},\{s'\}\} \text{ and }S\backslash\phi^{\M'}\notin\{\{s',t'\},\{s'\}\}\\
    \iff & \phi^{\M'}\neq \{s',t'\}\text{ and }\phi^{\M'}\neq\{s'\}\text{ and }\phi^{\M'}\neq\emptyset\text{ and }\phi^{\M'}\neq\{t'\}\\
    \iff & \text{false}\\
\end{array}
\]

In either case, the penultimate line of the equivalences states that $\phi$ cannot be interpreted on the related model: its denotation is not one of all
possible subsets of the domain. We therefore conclude that $\M,s\vDash\nabla\phi$ iff $\M',s'\vDash\nabla\phi$.

Finally, $(\M,s)$ and $(\M',s')$ can be distinguished by $\mathcal{L}(\Diamond)$. To see this, note that $p^\M=\emptyset \in N(s)$, thus $\M,s\vDash\Box p$; however, $p^{\M'}=\emptyset\notin N'(s')$, and thus $\M',s'\nvDash\Box p$.
\end{proof}

\begin{proposition}\label{prop.nablabullet-diamond-ct}
$\mathcal{L}(\nabla,\bullet)\equiv\mathcal{L}(\Box)[\mathbb{C}]$, where $\mathbb{C}\in\{\mathbb{C}_c,\mathbb{C}_t\}$.
%$\mathcal{L}(\nabla,\bullet)$ is equally expressive as $\mathcal{L}(\Diamond)$ over the class of neighborhood models satisfying $(c)$ or $(t)$.
\end{proposition}

\begin{proof}
Straightforward from $\mathcal{L}(\nabla,\bullet)\equiv\mathbb{L}(\nabla)[\mathbb{C}]$ (see Coro.~\ref{coro.exp-nabla-bullet-nablabullet}) and $\mathcal{L}(\nabla)\equiv\mathcal{L}(\Box)[\mathbb{C}]$~\cite[Prop.~5, Prop.~6]{FanvD:neighborhood}, where $\mathbb{C}\in\{\mathbb{C}_c,\mathbb{C}_t\}$.
%This is because $\mathcal{L}(\nabla,\bullet)$ has $\mathcal{L}(\nabla)$ as a fragment, and $\mathcal{L}(\nabla)$ is equally expressive as $\mathcal{L}(\Diamond)$ over the class of neighborhood models in question~\cite[Prop.~5, Prop.~6]{FanvD:neighborhood}.
\end{proof}

We do not know whether $\mathcal{L}(\nabla,\bullet)$ is less expressive than $\mathcal{L}(\Box)$ over the class of $(d)$-models. We conjecture the answer is positive. We leave it for future work.

We summarize the results in this section as follows.
\[
\begin{array}{ll}
    \mathcal{L}(\nabla)\asymp\mathcal{L}(\bullet)[\mathbb{C}],\text{ where }\mathbb{C}\in\{\mathbb{C}_\text{all},\mathbb{C}_n,\mathbb{C}_r,\mathbb{C}_i,\mathbb{C}_s,\mathbb{C}_d,\mathbb{C}_b,\mathbb{C}_4,\mathbb{C}_5\}& (\text{Coro.}~\ref{coro.exp-nabla-bullet-nablabullet}) \\
    \mathcal{L}(\nabla)\equiv\mathcal{L}(\bullet)[\mathbb{C}],\text{ where }\mathbb{C}\in\{\mathbb{C}_c,\mathbb{C}_t\}& (\text{Coro.}~\ref{coro.exp-nabla-bullet-nablabullet})\\
    \mathcal{L}(\nabla)\prec\mathcal{L}(\nabla,\bullet)[\mathbb{C}],\text{ where }\mathbb{C}\in\{\mathbb{C}_\text{all},\mathbb{C}_n,\mathbb{C}_r,\mathbb{C}_i,\mathbb{C}_s,\mathbb{C}_d,\mathbb{C}_b,\mathbb{C}_4,\mathbb{C}_5\} &(\text{Coro.}~\ref{coro.exp-nabla-bullet-nablabullet})\\
    \mathcal{L}(\bullet)\prec\mathcal{L}(\nabla,\bullet)[\mathbb{C}],\text{ where }\mathbb{C}\in\{\mathbb{C}_\text{all},\mathbb{C}_n,\mathbb{C}_r,\mathbb{C}_i,\mathbb{C}_s,\mathbb{C}_d,\mathbb{C}_b,\mathbb{C}_4,\mathbb{C}_5\} &(\text{Coro.}~\ref{coro.exp-nabla-bullet-nablabullet})\\
    \mathcal{L}(\nabla)\equiv\mathcal{L}(\nabla,\bullet)[\mathbb{C}],\text{ where }\mathbb{C}\in\{\mathbb{C}_c,\mathbb{C}_t\} &(\text{Coro.}~\ref{coro.exp-nabla-bullet-nablabullet})\\
    \mathcal{L}(\bullet)\equiv\mathcal{L}(\nabla,\bullet)[\mathbb{C}],\text{ where }\mathbb{C}\in\{\mathbb{C}_c,\mathbb{C}_t\} &(\text{Coro.}~\ref{coro.exp-nabla-bullet-nablabullet})\\
    \mathcal{L}(\nabla,\bullet)\prec\mathcal{L}(\Box)[\mathbb{C}],\text{ where }\mathbb{C}\in\{\mathbb{C}_\text{all},\mathbb{C}_n,\mathbb{C}_r,\mathbb{C}_i,\mathbb{C}_s,\mathbb{C}_b,\mathbb{C}_4,\mathbb{C}_5\} &(\text{Props.}~\ref{prop.exp-nablabullet-diamond-ri45},\ref{prop.exp-nablabullet-diamond-nsb})\\
    \mathcal{L}(\nabla,\bullet)\equiv\mathcal{L}(\Box)[\mathbb{C}],\text{ where }\mathbb{C}\in \{\mathbb{C}_c,\mathbb{C}_t\} &(\text{Coro.}~\ref{prop.nablabullet-diamond-ct})\\
\end{array}
\]

\section{Frame Definability}\label{sec.framedefinability}

\weg{\begin{proposition}
The frame properties $(n)$, $(r)$, $(i)$, $(s)$, $(c)$, $(d)$, $(t)$, $(b)$, $(4)$, and $(5)$ are undefinable in $\mathcal{L}(\nabla,\bullet)$.
\end{proposition}}
%This part investigates the frame definability of $\mathcal{L}(\nabla,\bullet)$.
We have shown in the previous section that $\mathcal{L}(\nabla,\bullet)$ is more expressive than $\mathcal{L}(\nabla)$ and $\mathcal{L}(\bullet)$ (at the level of models). It may then be natural to ask whether a similar situation holds at the level of frames. Recall that it is shown in~\cite[Prop.~7]{FanvD:neighborhood} that all frame properties in Def.~\ref{def.properties}, in particular $(n)$, are undefinable in $\mathcal{L}(\nabla)$. In what follows, we shall show that all frame properties in question except for $(n)$ are undefinable in $\mathcal{L}(\nabla,\bullet)$, thus $\mathcal{L}(\nabla,\bullet)$ is also more expressive than $\mathcal{L}(\nabla)$ and $\mathcal{L}(\bullet)$ at the level of frames. First, we need some related notion.
\begin{definition}
Let $\Gamma$ be a set of $\mathcal{L}(\nabla,\bullet)$-formulas, and $P$ a neighborhood property. We say that $\Gamma$ defines $P$, if for all frames $\mathcal{F}$, $\F$ has $P$ if and only if $\F\vDash\Gamma$. If $\Gamma$ is a singleton, say $\{\phi\}$, we will write $\F\vDash\phi$ rather than $\F\vDash\{\phi\}$. We say that $P$ is definable in $\mathcal{L}(\nabla,\bullet)$, if there exists a set of $\mathcal{L}(\nabla,\bullet)$-formulas that defines it.
\end{definition}

\begin{proposition}\label{prop.n-definable}
The frame property $(n)$ is definable in $\mathcal{L}(\nabla,\bullet)$.
\end{proposition}

\begin{proof}
\cite{Fan:2020neighborhood} has shown that $(n)$ is defined in $\mathcal{L}(\bullet)$, by $\circ\top$. Therefore, $(n)$ is also definable in $\mathcal{L}(\nabla,\bullet)$, by $\circ\top$.
\end{proof}

\begin{proposition}
The frame properties $(r)$, $(i)$, $(c)$, $(d)$, $(t)$ and $(b)$ are undefinable in $\mathcal{L}(\nabla,\bullet)$.
\end{proposition}

\begin{proof}
Consider the following frames $\mathcal{F}_1=\lr{S_1,N_1}$, $\mathcal{F}_2=\lr{S_2,N_2}$, and $\F_3=\lr{S_3,N_3}$\footnote{This come from~\cite[Prop.~7]{FanvD:neighborhood}.}:
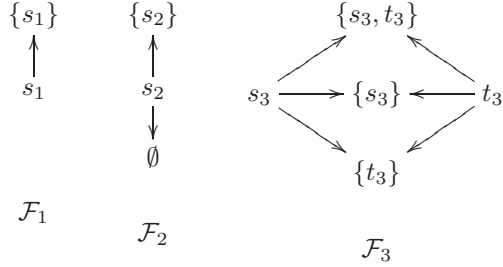
$$
%\hspace{-1cm}
\xymatrix@L-10pt@C-6pt@R-10pt{\{s_1\}\\
s_1\ar[u]\\
\empty\\
\mathcal{F}_1}
\qquad
\xymatrix@L-10pt@C-6pt@R-10pt{\{s_2\}\\
s_2\ar[u]\ar[d]\\
\emptyset\\
\mathcal{F}_2}
\qquad
\xymatrix@L-10pt@C-6pt@R-10pt{&\{s_3,t_3\}&\\
s_3\ar[ur]\ar[dr]\ar[r]&\{s_3\}&t_3\ar[ul]\ar[dl]\ar[l]\\
&\{t_3\}&\\
&\mathcal{F}_3&}
$$
It has been observed in~\cite[Prop.~7]{FanvD:neighborhood} that $\F_1$ satisfies $(d)$ and $(t)$ but $\F_2$ does not. Also, it is straightforward to check that $\mathcal{F}_2$ satisfies $(c)$ but $\mathcal{F}_1$ does not. Moreover, $\F_2$ satisfies $(r)$, $(i)$ and $(b)$, whereas $\F_3$ does not. To see $\F_3$ does not satisfy $(b)$, note that $s_3\in\{s_3\}$ but $\{u\in S_3\mid \{t_3\}\notin N_3(u)\}=\emptyset\notin N_3(s_3)$. In what follows, we show that for all $\phi\in\mathcal{L}(\nabla,\bullet)$, $\F_1\vDash\phi$ iff $\F_2\vDash\phi$ iff $\F_3\vDash\phi$.

Suppose that $\F_1\nvDash\phi$. Then there exists $\M_1=\lr{\F_1,V_1}$ such that $\M_1,s_1\nvDash\phi$. Define a valuation $V_2$ on $\F_2$ as $s_2\in V_2(p)$ iff $s_1\in V_1(p)$ for all $p\in\BP$. By induction on $\phi$, we show that $(\ast)$: $\M_1,s_1\vDash\phi$ iff $\M_2,s_2\vDash\phi$, where $\M_2=\lr{\F_2,V_2}$. The nontrivial cases are $\nabla\phi$ and $\bullet\phi$. The case $\nabla\phi$ can be shown as in~\cite[Prop.~7]{FanvD:neighborhood}. For the case $\bullet\phi$, notice that $\M_1,s_1\vDash\bullet\phi$ iff ($\M_1,s_1\vDash\phi$ and $\phi^{\M_1}\notin N_1(s_1)$) iff ($\M_1,s_1\vDash\phi$ and $\phi^{\M_1}\neq \{s_1\}$), where the last one is a contradiction, and thus $\M_1,s_1\nvDash\bullet\phi$; a similar argument gives us $\M_2,s_2\nvDash\bullet\phi$. We have thus proved $(\ast)$. This entails that $\M_2,s_2\nvDash\phi$, and thus $\F_2\nvDash\phi$. The converse is similar. Therefore, $\F_1\vDash\phi$ iff $\F_2\vDash\phi$.
%For the case $\nabla\phi$, note that $\M_1,s_1\vDash\nabla\phi$ iff ($\phi^{\M_1}\notin N_1(s_1)$ and $S_1\backslash\phi^{\M_1}\notin N_1(s_1)$) iff ($\phi^{\M_1}\neq \{s_1\}$ and $\phi^{\M_1}\neq \emptyset$) iff false, and $\M_2,s_2\vDash\nabla\phi$ iff ($\phi^{\M_2}\notin N_2(s_2)$ and $S_2\backslash\phi^{\M_2}\notin N_2(s_2)$) iff ($\phi^{\M_2}\neq \{s_2\}$ and $\phi^{\M_2}\neq \emptyset$) iff false). We have thus proved $(\ast)$. This entails that $\M_2,s_2\nvDash\phi$, and thus $\F_2\nvDash\phi$. The converse is similar. Therefore, $\F_1\vDash\phi$ iff $\F_2\vDash\phi$.

It remains only to show that $\F_2\vDash\phi$ iff $\F_3\vDash\phi$. Assume that $\F_2\nvDash\phi$. Then there exists $\M_2=\lr{\F_2,V_2}$ such that $\M_2,s_2\nvDash\phi$. Define a valuation $V_3$ on $\F_3$ such that $s_3\in V_3(p)$ iff $s_2\in V_2(p)$ for all $p\in \BP$. By induction on $\phi\in\mathcal{L}(\nabla,\bullet)$, we show that $(\ast\ast)$: $\M_2,s_2\vDash\phi$ iff $\M_3,s_3\vDash\phi$, where $\M_3=\lr{\F_3,V_3}$. The nontrivial cases are $\nabla\phi$ and $\bullet\phi$. Again, the case $\nabla\phi$ can be shown as in~\cite[Prop.~7]{FanvD:neighborhood}. For the case $\bullet\phi$, just note that $\M_3,s_3\vDash\bullet\phi$ iff ($\M_3,s_3\vDash\phi$ and $\phi^{\M_3}\notin N_3(s_3)$) iff ($\M_3,s_3\vDash\phi$ and $\phi^{\M_3}\neq \{s_3\}$ and $\phi^{\M_3}\neq \{t_3\}$ and $\phi^{\M_3}\neq \{s_3,t_3\}$) iff false. Thus $(\ast\ast)$ holds. This implies that $\M_3,s_3\nvDash\phi$, and then $\F_3\nvDash\phi$. The converse is analogous. Therefore, $\F_2\vDash\phi$ iff $\F_3\vDash\phi$.

If $(r)$ were to be defined by a set of $\mathcal{L}(\nabla,\bullet)$-formulas, say $\Sigma$, then as $\F_2$ satisfies $(r)$, we have $\F_2\vDash \Sigma$. Then we should also have $\F_3\vDash\Sigma$, which means that $\F_3$ has $(r)$: a contradiction. Therefore, $(r)$ is undefinable in $\mathcal{L}(\nabla,\bullet)$. Similarly, we can show other frame properties in question are undefinable in $\mathcal{L}(\nabla,\bullet)$.
%Similarly, we can show that $\F_2\vDash\phi$ iff $\F_3\vDash\phi$.
\end{proof}

\begin{proposition}\label{prop.undefinable-4}
The frame properties $(s)$ and $(4)$ are undefinable in $\mathcal{L}(\nabla,\bullet)$.
\end{proposition}

\begin{proof}
Consider the following frames $\mathcal{F}=\lr{S,N}$ and $\mathcal{F}'=\lr{S',N'}$, where $S=\{s,t\}$ and $S'=\{s',t'\}$:
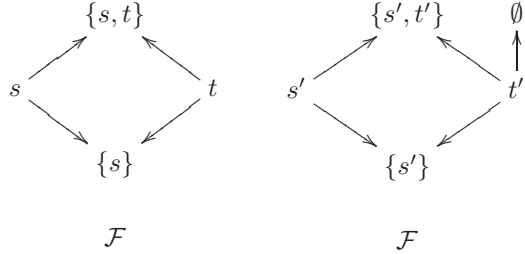
$$\xymatrix@L-10pt@C-6pt@R-10pt{&\{s,t\}&\\
s\ar[ur]\ar[dr]&&t\ar[ul]\ar[dl]\\
&\{s\}&\\
&\mathcal{F}&}
\qquad
\xymatrix@L-10pt@C-6pt@R-10pt{&\{s',t'\}&\emptyset\\
s'\ar[ur]\ar[dr]&&t'\ar[ul]\ar[dl]\ar[u]\\
&\{s'\}&\\
&\mathcal{F}&}$$

Firstly, one may easily see that $\mathcal{F}$ has $(s)$. Also, $\mathcal{F}$ has $(4)$.  Suppose that $X\in N(s)$, to show that $\{u\in S\mid X\in N(u)\}\in N(s)$. By supposition, $X=\{s\}$ or $X=\{s,t\}$. Either case implies that $\{u\in S\mid X\in N(u)\}=\{s,t\}\in N(s)$. Thus $N(s)$ has $(4)$. A similar argument applies to showing that $N(t)$ has $(4)$.

Secondly, $\F'$ does not have $(s)$, since $\emptyset\in N'(t')$ and $\emptyset\subseteq \{t'\}$ but $\{t'\}\notin N'(t')$. Moreover, $\F'$ does not have $(4)$. This is because, for instance, $\emptyset\in N'(t')$ but $\{u\in S'\mid \emptyset\in N'(u)\}=\{t'\}\notin N'(t')$.

Thirdly, for all $\psi\in\mathcal{L}(\nabla,\bullet)$, we have that $\F\vDash\psi$ iff $\F'\vDash\psi$. Suppose that $\F\nvDash \psi$. Then there exists $\M=\lr{\F,V}$ and $x\in S$ such that $\M,x\nvDash \psi$. Define $V'$ to be a valuation on $\F'$ such that $s\in V(p)$ iff $s'\in V'(p)$, and $t\in V(p)$ iff $t'\in V'(p)$. In what follows, we show $(\ast)$: for all $\phi\in\mathcal{L}(\nabla,\bullet)$, $\M,s\vDash\phi$ iff $\M',s'\vDash\phi$, and $\M,t\vDash\phi$ iff $\M',t'\vDash\phi$, where $\M'=\lr{\F',N'}$. We proceed by induction on $\phi$, where the nontrivial cases are $\nabla\phi$ and $\bullet\phi$. 

For the case $\nabla\phi$, we have the following equivalences.
\[
\begin{array}{ll}
     & \M,s\vDash\nabla\phi \\
   \iff  & \phi^\M\notin N(s)\text{ and }S\backslash\phi^\M\notin N(s)\\
   \iff & \phi^\M\notin \{\{s\},\{s,t\}\}\text{ and }S\backslash\phi^\M\notin\{\{s\},\{s,t\}\}\\
   \iff & \phi^\M\neq \{s\}\text{ and }\phi^\M\neq \{s,t\}\text{ and }\phi^\M\neq \{t\}\text{ and }\phi^\M\neq \emptyset\\
\end{array}
\]
\[
\begin{array}{ll}
     & \M',s'\vDash\nabla\phi \\
   \iff  & \phi^{\M'}\notin N'(s')\text{ and }S'\backslash\phi^{\M'}\notin N'(s')\\
   \iff & \phi^{\M'}\notin \{\{s'\},\{s',t'\}\}\text{ and }S'\backslash\phi^{\M'}\notin\{\{s'\},\{s',t'\}\}\\
   \iff & \phi^{\M}\neq \{s'\}\text{ and }\phi^{\M'}\neq \{s',t'\}\text{ and }\phi^{\M'}\neq \{t'\}\text{ and }\phi^{\M'}\neq \emptyset\\
\end{array}
\]

In each case, the last line of the above proofs states that $\phi$ cannot be interpreted on the related models, which is impossible. Thus $\M,s\nvDash\nabla\phi$ and $\M',s'\nvDash\nabla\phi$. Analogously, we can show that $\M,t\nvDash\nabla\phi$ and $\M',t'\nvDash\nabla\phi$.

For the case $\bullet\phi$, we have the following equivalences.
\[
\begin{array}{ll}
     & \M,s\vDash\bullet\phi \\
    \iff  & \M,s\vDash\phi\text{ and }\phi^\M\notin N(s)\\
    \iff & \M,s\vDash\phi\text{ and }\phi^\M\neq \{s\}\text{ and }\phi^\M\neq \{s,t\}\\
    \iff & \text{false}\\
\end{array}\]
\[
\begin{array}{ll}
     & \M',s'\vDash\bullet\phi \\
    \iff  & \M',s'\vDash\phi\text{ and }\phi^{\M'}\notin N'(s')\\
    \iff & \M',s'\vDash\phi\text{ and }\phi^{\M'}\neq \{s'\}\text{ and }\phi^{\M'}\neq \{s',t'\}\\
    \iff & \text{false}\\
\end{array}
\]

This shows that $\M,s\vDash\bullet\phi$ iff $\M',s'\vDash\bullet\phi$. Therefore, $\M,s\vDash\phi$ iff $\M',s'\vDash\phi$ for all $\phi\in\mathcal{L}(\nabla,\bullet)$.
\[
\begin{array}{ll}
     & \M,t\vDash\bullet\phi \\
    \iff & \M,t\vDash\phi\text{ and }\phi^\M\notin N(t)\\
    \iff & \M,t\vDash\phi\text{ and }\phi^\M\neq \{s\}\text{ and }\phi^\M\neq \{s,t\}\\
    \iff & \M,t\vDash\phi\text{ and }\M,s\nvDash\phi\\
   \stackrel{\text{IH}}{\iff} & \M',t'\vDash\phi\text{ and }\M',s'\nvDash\phi\\
   \iff & \M',t'\vDash\phi\text{ and }\phi^{\M'}\neq \{s'\}\text{ and }\phi^{\M'}\neq \{s',t'\}\text{ and }\phi^{\M'}\neq \emptyset\\
   \iff & \M',t'\vDash\phi\text{ and }\phi^{\M'}\notin N'(t')\\
   \iff & \M',t'\vDash\bullet\phi\\
    \end{array}
\]

This gives us that $\M,t\vDash \phi$ iff $\M',t'\vDash\phi$ for all $\phi\in\mathcal{L}(\nabla,\bullet)$. We have now completed the proof of $(\ast)$.

If $(s)$ were to be defined by a set of $\mathcal{L}(\nabla,\bullet)$-formulas, say $\Gamma$, then as $\F$ has $(s)$, we would have $\F\vDash\Gamma$, thus we should also have $\F'\vDash\Gamma$, that is, $\F'$ has $(s)$: a contradiction. Therefore, $(s)$ is not definable in $\mathcal{L}(\nabla,\bullet)$. Similarly, we can obtain the undefinability of $(4)$ in $\mathcal{L}(\nabla,\bullet)$.
\end{proof}

%Note that Prop.~\ref{prop.undefinable-4} also shows that $(n)$ is undefinable in $\mathcal{L}(\nabla,\bullet)$, since $\F$ has $(n)$ but $\F'$ fails.

\begin{proposition}
The frame property $(5)$ is undefinable in $\mathcal{L}(\nabla,\bullet)$.
\end{proposition}

\begin{proof}
Consider the following frames $\mathcal{F}=\lr{S,N}$ and $\mathcal{F}'=\lr{S',N'}$, where $S=\{s,t\}$ and $S'=\{s',t'\}$:
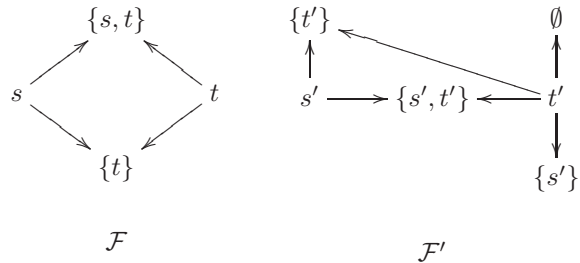
$$\xymatrix@L-10pt@C-6pt@R-10pt{&\{s,t\}&\\
s\ar[ur]\ar[dr]&&t\ar[ul]\ar[dl]\\
&\{t\}&\\
&\mathcal{F}&}
\qquad
\xymatrix@L-10pt@C-6pt@R-10pt{\{t'\}&&\emptyset\\
s'\ar[u]\ar[r]&\{s',t'\}&t'\ar[u]\ar[d]\ar[l]\ar[ull]\\
&&\{s'\}\\
&\mathcal{F}'&}$$

Firstly, $\F$ has $(5)$. Suppose that $X\notin N(s)$, to prove that $\{u\in S\mid X\notin N(u)\}\in N(s)$. By supposition, $X=\emptyset$ or $X=\{s\}$. Either case implies that $\{u\in S\mid X\notin N(u)\}=\{s,t\}\in N(s)$. Thus $N(s)$ has $(5)$. A similar argument shows that $N(t)$ has $(5)$.

Secondly, $\F'$ does not $(5)$. For instance, $\emptyset\notin N'(s')$ and $\{u\in S'\mid \emptyset\notin N'(u)\}=\{s'\}\notin N'(s')$.

Thirdly, for all $\psi\in\mathcal{L}(\nabla,\bullet)$, we have that $\F\vDash\psi$ iff $\F'\vDash\psi$. Suppose that $\F\nvDash \psi$. Then there exists $\M=\lr{\F,V}$ and $x\in S$ such that $\M,x\nvDash \psi$. Define $V'$ to be a valuation on $\F'$ such that $s\in V(p)$ iff $s'\in V'(p)$, and $t\in V(p)$ iff $t'\in V'(p)$. In what follows, we show $(\ast\ast)$: for all $\phi\in\mathcal{L}(\nabla,\bullet)$, $\M,s\vDash\phi$ iff $\M',s'\vDash\phi$, and $\M,t\vDash\phi$ iff $\M',t'\vDash\phi$, where $\M'=\lr{\F',N'}$. We proceed by induction on $\phi$, where the nontrivial cases are $\nabla\phi$ and $\bullet\phi$. 

%We have seen from the proof of Prop.~\ref{prop.undefinable-4} that $\M,s\nvDash\nabla\phi$, $\M,t\nvDash\nabla\phi$, $\M,t\nvDash\bullet\phi$, and $\M,s\vDash\bullet\phi$ iff $\M,s\vDash\phi$. Moreover, we have the following equivalences.
For the case $\nabla\phi$, we have the following equivalences.
\[
\begin{array}{ll}
     & \M,s\vDash\nabla\phi \\
   \iff  & \phi^\M\notin N(s)\text{ and }S\backslash\phi^\M\notin N(s)\\
   \iff & \phi^\M\notin \{\{t\},\{s,t\}\}\text{ and }S\backslash\phi^\M\notin\{\{t\},\{s,t\}\}\\
   \iff & \phi^\M\neq \{t\}\text{ and }\phi^\M\neq \{s,t\}\text{ and }\phi^\M\neq \{s\}\text{ and }\phi^\M\neq \emptyset\\
\end{array}
\]
\[
\begin{array}{ll}
     & \M',s'\vDash\nabla\phi \\
    \iff & \phi^{\M'}\notin N'(s')\text{ and }S'\backslash\phi^{\M'}\notin N'(s')\\
    \iff & \phi^{\M'}\notin \{\{t'\},\{s',t'\}\}\text{ and }S'\backslash\phi^{\M'}\notin \{\{t'\},\{s',t'\}\}\\
    \iff & \phi^{\M'}\neq \{t'\}\text{ and }\phi^{\M'}\neq \{s',t'\}\text{ and }\phi^{\M'}\neq \{s'\}\text{ and }\phi^{\M'}\neq \emptyset\\
\end{array}
\]

In each case, the last line of the above proofs states that $\phi$ cannot be interpreted on the related models, which is impossible. Thus $\M,s\nvDash\nabla\phi$ and $\M',s'\nvDash\nabla\phi$. Analogously, we can show that $\M,t\nvDash\nabla\phi$ and $\M',t'\nvDash\nabla\phi$.

%The last line of the above proof states that $\phi$ cannot be interpreted on $\M'$, which is impossible. Thus $\M',s'\nvDash\nabla\phi$. Similarly, we can show that $\M',t'\nvDash\nabla\phi$.

For the case $\bullet\phi$, we have the following equivalences.
\[
\begin{array}{ll}
     & \M,t\vDash\bullet\phi \\
    \iff & \M,t\vDash\phi\text{ and }\phi^\M\notin N(t)\\
    \iff & \M,t\vDash\phi\text{ and }\phi^{\M}\neq \{t\}\text{ and }\phi^{\M}\neq \{s,t\}\\
    \iff & \text{false}\\
\end{array}
\]
\[
\begin{array}{ll}
     & \M',t'\vDash\bullet\phi \\
    \iff & \M',t'\vDash\phi\text{ and }\phi^{\M'}\notin N'(t')\\
    \iff & \M',t'\vDash\phi\text{ and }\phi^{\M'}\neq \emptyset\text{ and }\phi^{\M'}\neq \{s'\}\text{ and }\phi^{\M'}\neq \{t'\}\text{ and }\phi^{\M'}\neq \{s',t'\}\\
    \iff & \text{false}\\
\end{array}
\]

Thus $\M,t\vDash\bullet\phi$ iff $\M',t'\vDash\bullet\phi$. Therefore, $\M,t\vDash\phi$ iff $\M',t'\vDash\phi$ for all $\phi\in\mathcal{L}(\nabla,\bullet)$.
\[
\begin{array}{ll}
     & \M,s\vDash\bullet\phi \\
    \iff  & \M,s\vDash\phi\text{ and }\phi^\M\notin N(s)\\
    \iff & \M,s\vDash\phi\text{ and }\phi^\M\neq \{t\}\text{ and }\phi^\M\neq \{s,t\}\\
    \iff & \M,s\vDash\phi\text{ and }\M,t\nvDash\phi\\
    \stackrel{\text{IH}}{\iff} & \M',s'\vDash\phi\text{ and }\M',t'\nvDash\phi\\
    \iff & \M',s'\vDash\phi\text{ and }\phi^{\M'}\neq \{t'\}\text{ and }\phi^{\M'}\neq \{s',t'\}\\
    \iff & \M',s'\vDash\phi\text{ and }\phi^{\M'}\notin N'(s')\\
    \iff & \M',s'\vDash\bullet\phi\\
\end{array}
\]

Therefore, $\M,s\vDash\phi$ iff $\M',s'\vDash\phi$ for all $\phi\in\mathcal{L}(\nabla,\bullet)$. This completes the proof of $(\ast\ast)$.

\weg{It remains only to show that $\M',s'\vDash\bullet\phi$ iff $\M,s'\vDash\phi$, and $\M',t'\nvDash\bullet\phi$. For the latter, note that $\M',t'\vDash\bullet\phi$ would imply that $\phi^{\M'}\notin N'(t')$, which then entails that $\phi$ cannot be interpreted on $\M'$: a contradiction. For the former, we have the following equivalences.
\[
\begin{array}{ll}
     & \M',s'\vDash\bullet\phi \\
    \iff & \M',s'\vDash\phi\text{ and }\phi^{\M'}\notin N'(s')\\
    \iff & \M',s'\vDash\phi\text{ and }(\text{either }\{s'\}\notin N'(s')\text{ or }\{s',t'\}\notin N'(s'))\\
    \iff & \M',s'\vDash\phi\\
\end{array}
\]
}
If $(5)$ were to be defined by a set of $\mathcal{L}(\nabla,\bullet)$-formulas, say $\Gamma$, then as $\F$ has $(5)$, we would have $\F\vDash\Gamma$, thus we should also have $\F'\vDash\Gamma$, that is, $\F'$ has $(5)$: a contradiction. Therefore, $(5)$ is not definable in $\mathcal{L}(\nabla,\bullet)$.
\end{proof}

\section{Axiomatizations}\label{sec.axiomatizations}

In this section, we axiomatize $\mathcal{L}(\nabla,\bullet)$ over various classes of neighborhood frames.

\subsection{Classical logic}
\subsubsection{Proof system and Soundness}

\begin{definition}
The classical logic of $\mathcal{L}(\nabla,\bullet)$, denoted ${\bf E^{\nabla\bullet}}$, consists of the following axioms and inference rules:
\[
\begin{array}{ll}
    \text{TAUT} & \text{all instances of tautologies}\\
    \text{E1} & \nabla\phi\lra \nabla\neg\phi \\
    \text{E2} & \bullet\phi\to\phi\\
    \text{E3} & \nabla\phi\to\bullet\phi\vee\bullet\neg\phi\\
    \text{MP} & \dfrac{\phi,\phi\to\psi}{\psi}\\
    \text{RE}\nabla& \dfrac{\phi\lra\psi}{\nabla\phi\lra\nabla\psi}\\
    \text{RE}\bullet & \dfrac{\phi\lra\psi}{\bullet\phi\lra\bullet\psi}\\
\end{array}
\]
\end{definition}

Intuitively, $\text{E1}$ says that one is (first-order) ignorant whether a proposition holds if and only if one is ignorant whether its negation holds; $\text{E2}$ says that one is (Fitchean) ignorant of the fact that $\phi$ only if it is the case that $\phi$; $\text{E3}$ describes the relationship between Fitchean ignorance and first-order ignorance: if one is ignorant whether $\phi$, then either one is ignorant of the fact that $\phi$ or one is ignorant of the fact that $\phi$ is not the case; $\text{RE}\nabla$ and $\text{RE}\bullet$ concerns the replacement of equivalences for first-order ignorance and Fitchean ignorance, respectively.

It is straightforward by axiom $\text{E2}$ that $\bullet\nabla\phi\to\nabla\phi$ is provable in ${\bf E^{\nabla\bullet}}$, which says that under any condition, Rumsfeld ignorance implies first-order ignorance.

The following result states how to derive Fitchean ignorance, which means that if one is {\em ignorant whether} a {\em true} proposition holds, then one is {\em ignorant of} the proposition. It will be used in several places below (for instance, Lemma~\ref{lem.truthlemma-e},  Lemma~\ref{lem.truthlemma-ec}, and Prop.~\ref{prop.impliescirctop}).
\begin{proposition}\label{prop.circimpliesdelta}
$\vdash\nabla\phi\land\phi\to\bullet\phi$. Equivalently,
$\vdash\circ\phi\land\phi\to\Delta\phi$.
\end{proposition}

\begin{proof}
We have the following proof sequence:
\[
\begin{array}{lll}
    (i) & \nabla\phi\to\bullet\phi\vee\bullet\neg\phi & \text{E3} \\
    (ii) & \bullet\neg\phi\to\neg\phi & \text{E2}\\
    (iii)& \nabla\phi\to\bullet\phi\vee\neg\phi& (i),(ii)\\
    (iv) & \nabla\phi\land\phi\to\bullet\phi& (iii)\\
    %(iv) & \circ\phi\land\phi\to\Delta\phi& (iii),\text{Def.~}\circ,\text{Def.~}\Delta\\
\end{array}
\]
\end{proof}

As a corollary, we have $\vdash \nabla\nabla\phi\land\nabla\phi\to\bullet\nabla\phi$, and thus $\vdash\nabla\nabla\phi\land\nabla\phi\to\bullet\nabla\phi$. This means, in Fine~\cite{Fine:2018} terms, second-order ignorance plus first-order ignorance implies Rumsfeld ignorance. On one hand, this is not noticed in Fine~\cite{Fine:2018}; on the other hand, this plus the transitivity entails that second-order ignorance implies Rumsfeld ignorance, a result in the paper in question.

The following result indicates how to derive a proposition from Fitchean ignorance.
\begin{proposition}\label{prop.impliesphi}
$\vdash\bullet(\circ \phi\vee\psi\to \phi)\to \phi$
\end{proposition}

\begin{proof}
We have the following proof sequence.
\[
\begin{array}{lll}
    (1) & \bullet(\circ \phi\vee\psi\to \phi)\to (\circ\phi\vee\psi\to \phi) & \text{E2} \\
    (2) & (\circ\phi\vee\psi\to \phi) \to (\circ\phi\to \phi) & \text{TAUT}\\
    (3) & \bullet(\circ \phi\vee\psi\to \phi)\to (\circ\phi\to\phi) & (1),(2)\\
    (4) & \bullet \phi\to \phi &\text{E2}\\
    (5) & \bullet(\circ \phi\vee\psi\to \phi)\to \phi & (3),(4)\\
\end{array}
\]
\end{proof}

\begin{proposition}\label{prop.e-soundness}
${\bf E^{\nabla\bullet}}$ is sound with respect to the class of all (neighborhood) frames.
\end{proposition}

\begin{proof}
It suffices to show the validity of axiom $\text{E3}$. This has been shown in the proof of Prop.~\ref{prop.exp-bullet-nabla-c}.
%Let $\M=\lr{S,N,V}$ be any model and $s\in S$. Suppose that $\M,s\vDash\nabla\phi$. Then $\phi^\M\notin N(s)$ and $S\backslash\phi^\M\notin N(s)$, that is, $(\neg\phi)^\M\notin N(s)$. We have either $\M,s\vDash\phi$ or $\M,s\vDash\neg\phi$. If $\M,s\vDash\phi$, since $\phi^\M\notin N(s)$, we infer that $\M,s\vDash\bullet\phi$; if $\M,s\vDash\neg\phi$, since $(\neg\phi)^\M\notin N(s)$, we derive that $\M,s\vDash{\bullet\neg}\phi$. Therefore, $\M,s\vDash\bullet\phi\vee{\bullet\neg}\phi$. Since $(\M,s)$ is arbitrary, this establishes the validity of axiom $\text{E3}$.
\end{proof}

\subsubsection{Completeness}

\begin{definition}\label{def.e-cm}
The canonical model for ${\bf E^{\nabla\bullet}}$ is $\M^c=\lr{S^c,N^c,V^c}$, where
\begin{itemize}
    \item[-] $S^c=\{s\mid s\text{ is a maximal consistent set for }{\bf E^{\nabla\bullet}}\}$,
    \item[-] $N^c(s)=\{|\phi|\mid \circ\phi\land\Delta\phi\in s\}$,
    \item[-] $V^c(p)=\{s\in S^c\mid p\in s\}$.
\end{itemize}
\end{definition}

\begin{lemma}\label{lem.truthlemma-e}
For all $\phi\in\mathcal{L}(\nabla,\bullet)$, for all $s\in S^c$, we have
$$\M^c,s\vDash \phi\iff \phi\in s.$$
That is, $|\phi|=\phi^{\M^c}$.
\end{lemma}
\begin{proof}
By induction on $\phi$. The nontrivial cases are $\nabla\phi$ and $\bullet\phi$.

For case $\nabla\phi$:

First, suppose that $\nabla\phi\in s$, to show that $\M^c,s\vDash\nabla\phi$. By supposition, we have $\Delta\phi\notin s$. Then by definition of $N^c$, we infer that $|\phi|\notin N^c(s)$. By supposition again and axiom \text{E1}, we have $\nabla\neg\phi\in s$, and thus $\Delta\neg\phi\notin s$, and hence $|\neg\phi|\notin N^c(s)$, that is, $S^c\backslash|\phi|\notin N^c(s)$. By induction hypothesis, we have $\phi^{\M^c}\notin N^c(s)$ and $S^c\backslash\phi^{\M^c}\notin N^c(s)$. Therefore, $\M^c,s\vDash\nabla\phi$.

Conversely, assume that $\nabla\phi\notin s$ (that is, $\nabla\neg\phi\notin s$), to show that $\M^c,s\nvDash\nabla\phi$. By assumption, $\Delta\phi\in s$ and $\Delta\neg\phi\in s$. Since $s\in S^c$, we have either $\phi\in s$ or $\neg\phi\in s$. If $\phi\in s$, then by axiom \text{E2}, $\circ\neg\phi\in s$, and then $|\neg\phi|\in N^c(s)$, viz. $S^c\backslash |\phi|\in N^c(s)$, which by induction hypothesis implies that $S^c\backslash\phi^{\M^c}\in N^c(s)$. If $\neg\phi\in s$, then again by axiom $\text{E2}$, $\circ\phi\in s$, thus $|\phi|\in N^c(s)$, which by induction hypothesis entails that $\phi^{\M^c}\in N^c(s)$. We have now shown that either $\phi^{\M^c}\in N^c(s)$ or $S^c\backslash\phi^{\M^c}\in N^c(s)$, and we therefore conclude that $\M^c,s\nvDash\nabla\phi$.

For case $\bullet\phi$:

First, suppose that $\bullet\phi\in s$, to show that $\M^c,s\vDash\bullet\phi$. By supposition and axiom \text{E2}, we have $\phi\in s$. By induction hypothesis, $\M^c,s\vDash\phi$. By supposition and definition of $N^c$, we infer that $|\phi|\notin N^c(s)$, which by induction means that $\phi^{\M^c}\notin N^c(s)$. Therefore, $\M^c,s\vDash\bullet\phi$.

Conversely, assume that $\bullet\phi\notin s$, to demonstrate that $\M^c,s\nvDash\bullet\phi$. By assumption, $\circ\phi\in s$. If $\M^c,s\nvDash\phi$, it is obvious that $\M^c,s\nvDash\bullet\phi$. Otherwise, by induction hypothesis, we have $\phi\in s$, then $\circ\phi\land\phi\in s$. By Prop.~\ref{prop.circimpliesdelta}, $\Delta\phi\in s$, and thus $|\phi|\in N^c(s)$, by induction we obtain $\phi^{\M^c}\in N^c(s)$, and therefore we have also $\M^c,s\nvDash\bullet\phi$.
\end{proof}

It is then a standard exercise to show the following.
\begin{theorem}\label{thm.e-comp}
${\bf E^{\nabla\bullet}}$ is sound and strongly complete with respect to the class of all neighborhood frames.
\end{theorem}

\subsection{Extensions}

\subsubsection{${\bf E_c^{\nabla\bullet}}$}

Define ${\bf E_c^{\nabla\bullet}}$ be the smallest extension of ${\bf E^{\nabla\bullet}}$ with the following axiom, denoted $\text{E4}$:
$$\bullet\phi\to \nabla\phi.$$

Intuitively, $\text{E4}$ says that Fitchean ignorance implies first-order ignorance.

From $\text{E4}$ we can easily prove $\Delta\phi\to\circ\phi$. This turns the canonical model for ${\bf E^{\nabla\bullet}}$ (Def.~\ref{def.e-cm}) into the following simpler one.

\begin{definition}
The canonical model for ${\bf E_c^{\nabla\bullet}}$ is a triple $\M^c=\lr{S^c,N^c,V^c}$, where
\begin{itemize}
\item $S^c=\{s\mid s\text{ is a maximal consistent set for }{\bf E_c^{\nabla\bullet}}\}$
\item $N^c(s)=\{|\phi|\mid \Delta\phi\in s\}$
\item $V^c(p)=\{s\in S^c\mid p\in s\}$.
\end{itemize}
\end{definition}

\begin{lemma}\label{lem.truthlemma-ec}
For all $\phi\in \mathcal{L}(\nabla,\bullet)$, for all $s\in S^c$, we have
$$\M^c,s\vDash\phi\iff \phi\in s.$$
That is, $|\phi|=\phi^{\M^c}$.
\end{lemma}

\begin{proof}
By induction on $\phi$. The nontrivial cases are $\nabla\phi$ and $\bullet\phi$. The case $\nabla\phi$ has been shown in~\cite[Lemma~1]{FanvD:neighborhood}. It suffices to show the case $\bullet\phi$.

Suppose that $\bullet\phi\in s$. Then by axiom $\text{E2}$, we have $\phi\in s$; by axiom $\text{E4}$, we derive that $\nabla\phi\in s$, and thus $\Delta\phi\notin s$. This follows that $|\phi|\notin N^c(s)$. By induction hypothesis, $\M^c,s\vDash\phi$ and $\phi^{\M^c}\notin N^c(s)$. Therefore, $\M^c,s\vDash\bullet\phi$.

Conversely, assume that $\bullet\phi\notin s$, to show that $\M^c,s\nvDash\bullet\phi$, which by induction hypothesis amounts to showing that $\phi\notin s$ or $|\phi|\in N^c(s)$. For this, suppose that $\phi\in s$, this plus the assumption implies that $\circ\phi\land\phi\in s$. Then by Prop.~\ref{prop.circimpliesdelta}, $\Delta\phi\in s$, and therefore $|\phi|\in N^c(s)$.
\end{proof}

\begin{proposition}
$\M^c$ possesses the property $(c)$.
\end{proposition}

\begin{proof}
Refer to~\cite[Thm.~2]{FanvD:neighborhood}.
\end{proof}

Now it is a standard exercise to show the following.
\begin{theorem}
${\bf E_c^{\nabla\bullet}}$ is sound and strongly complete with respect to the class of $(c)$-frames.
\end{theorem}

In the neighborhood context $(c)$, there is some relationship between Rumsfeld ignorance, second-order ignorance and first-order ignorance. The following is immediate from the axiom $\text{E4}$.

\begin{proposition}
$\bullet\nabla\phi\to\nabla\nabla\phi$ is provable in ${\bf E^{\nabla\bullet}_c}$.
\end{proposition}

This says that under the condition $(c)$, Rumsfeld ignorance implies second-order ignorance.

Combined with an instance of the axiom $\text{E2}$ ($\bullet\nabla\phi\to\nabla\phi$) and $\vdash\nabla\nabla\phi\land\nabla\phi\to\bullet\nabla\phi$ (see the remark after Prop.~\ref{prop.circimpliesdelta}), it follows that within the neighborhood context $(c)$, Rumsfeld ignorance amounts to second-order ignorance plus first-order ignorance, and thus Rumsfeld ignorance is definable in terms of first-order ignorance. 

\subsubsection{${\bf EN^{\nabla\bullet}}$}\label{sec.en}

%Let ${\bf EN^{\nabla\bullet}}={\bf E^{\nabla\bullet}}+\circ\top+\Delta\top$. Then we have

Let ${\bf EN^{\nabla\bullet}}={\bf E^{\nabla\bullet}}+\circ\top$. From $\circ\top$ and Prop.~\ref{prop.circimpliesdelta} it follows that $\Delta\top$ is derivable in ${\bf EN^{\nabla\bullet}}$.
%One may ask why we do not add the formula $\Delta\top$ as an extra axiom. This is because this formula is derivable from axiom $\text{K1}$ and Prop.~\ref{prop.circimpliesdelta}. 
\begin{theorem}
${\bf EN^{\nabla\bullet}}$ is sound and strongly complete with respect to the class of all $(n)$-frames.
\end{theorem}

\begin{proof}
For soundness, by Prop.~\ref{prop.e-soundness}, it remains only to show the validity of $\circ\top$ over the class of $(n)$-frames. The validity of $\circ\top$ can be found in Prop.~\ref{prop.n-definable}.

For completeness, define the canonical model $\M^c$ w.r.t. ${\bf EN^{\nabla\bullet}}$ as in Def.~\ref{def.e-cm}. By Thm.~\ref{thm.e-comp}, it suffices to show that $\M^c$ possesses $(n)$. By the construction of ${\bf EN^{\nabla\bullet}}$, for all $s\in S^c$, we have $\circ\top\land\Delta\top\in s$, and thus $|\top|\in N^c(s)$, that is, $S^c\in N^c(s)$, as desired.
\end{proof}

\weg{\subsubsection{${\bf EC^{\nabla\bullet}}$}

Let ${\bf EC^{\nabla\bullet}}={\bf E^{\nabla\bullet}}+\Delta\text{C}+\circ\text{C}$, where
\[\begin{array}{ll}
\text{C1} & \nabla(\phi\land\psi)\to \nabla\phi\vee\nabla\psi\\
%\Delta\phi\land\Delta\psi\to\Delta(\phi\land\psi) \\
\text{C2} & \bullet(\phi\land\psi)\to \bullet\phi\vee\bullet\psi\\
%\circ\phi\land\circ\psi\to\circ(\phi\land\psi)\\
\end{array}\]
}

\subsubsection{Monotone logic}

Let ${\bf M^{\nabla\bullet}}$ be the extension of ${\bf E^{\nabla\bullet}}$ plus the following extra axioms:
%$\bullet(\phi\to\psi)\land\bullet(\neg\phi\to\chi)\to\nabla\phi$.
%$\Delta\phi\to\Delta(\phi\vee\psi)\vee\Delta(\neg\phi\vee\chi)$
\[\begin{array}{ll}
    \text{M1} & \nabla(\phi\vee\psi)\land\nabla(\neg\phi\vee\chi)\to\nabla\phi \\
    \text{M2} & \bullet(\phi\vee\psi)\land\bullet(\neg\phi\vee\chi)\to\nabla\phi\\
    \text{M3} & \bullet(\phi\vee\psi)\land\nabla(\neg\phi\vee\chi)\to\nabla\phi\\
     \text{M4} & \circ\phi\land\phi\to\Delta(\phi\vee\psi)\land\circ(\phi\vee\psi)\\
\end{array}\]

Prop.~\ref{prop.impliesdeltatop} and Prop.~\ref{prop.impliescirctop} tells us how to derive $\Delta\top$ and $\circ\top$ in ${\bf M^{\nabla\bullet}}$, respectively. They will used in Section~\ref{sec.updating}.
\begin{proposition}\label{prop.impliesdeltatop}
$\Delta \phi\to\Delta \top$ is provable in ${\bf M^{\nabla\bullet}}$.
\end{proposition}

\begin{proof}
We have the following proof sequence in ${\bf M^{\nabla\bullet}}$.
\[
\begin{array}{lll}
    (1) & \nabla(\phi\vee \top)\land \nabla (\neg\phi\vee\top)\to \nabla\phi & \text{M1} \\
    (2) & \top\lra \phi\vee\top & \text{TAUT}\\
    (3) & \nabla\top\lra \nabla(\phi\vee\top) & (2),\text{RE}\nabla\\
    (4) & \top \lra \neg\phi\vee\top & \text{TAUT}\\
    (5) & \nabla\top\lra \nabla(\neg\phi\vee\top) & (4),\text{RE}\nabla\\
    (6) & \nabla\top\to\nabla\phi & (1),(3),(5)\\
    (7) & \Delta \phi\to \Delta\top & (6),\text{Def.~}\Delta\\
\end{array}
\]
\end{proof}

In the above proof, $\nabla\top\to\nabla\phi$ says that if one is ignorant about whether $\top$ holds, then one is ignorant about whether everything holds.

\begin{proposition}\label{prop.impliescirctop}
$\phi\land\circ\phi\to\circ\top$ is provable in ${\bf M^{\nabla\bullet}}$.
\end{proposition}

\begin{proof}
We have the following proof sequence in ${\bf M^{\nabla\bullet}}$.
\[
\begin{array}{lll}
    (1) & (\phi\vee\top)\lra \top & \text{TAUT}\\
    (2) & \bullet(\phi\vee\top)\lra \bullet\top & (1),\text{RE}\bullet\\
    (3) & (\neg\phi\vee\top)\lra \top & \text{TAUT}\\
    (4) & \bullet(\neg\phi\vee\top)\lra \bullet\top & (3),\text{RE}\bullet\\
    (5) & \bullet(\phi\vee\top)\land\bullet(\neg\phi\vee\top)\to \nabla\phi & \text{M2}\\
    (6) & \bullet\top\to \nabla\phi & (2),(4),(5)\\
    (7) & \Delta\phi\to\circ\top & (6),\text{Def.~}\Delta,\text{Def.~}\circ\\
    (8) & \phi\land\circ\phi\to\Delta\phi & \text{Prop.~}\ref{prop.circimpliesdelta} \\
    (9) & \phi\land\circ\phi\to\circ\top & (7),(8)\\
\end{array}
\]
\weg{\[
\begin{array}{lll}
    (1) & \phi\land\circ\phi\to \circ(\phi\vee\top) & \text{M4} \\
    (2) & \phi\vee\top \lra \top & \text{TAUT}\\
    (3) & \bullet(\phi\vee\top)\lra \bullet\top & (2),\text{RE}\bullet\\
    (4) & {\neg\bullet}(\phi\vee\top)\lra {\neg\bullet}\top & (3)\\
    (5) & \circ(\phi\vee\top)\lra \circ\top & (4),\text{Def.~}\circ\\
    (6) & \phi\land\circ\phi\to\circ\top & (1),(5)\\
\end{array}
\]}
\end{proof}

\begin{proposition}
${\bf M^{\nabla\bullet}}$ is sound with respect to the class of all $(s)$-frames.
\end{proposition}

\begin{proof}
By soundness of ${\bf E^{\nabla\bullet}}$ (Prop.~\ref{prop.e-soundness}), it suffices to show the validity of the extra axioms. Let $\M=\lr{S,N,V}$ be an arbitrary $(s)$-model and $s\in S$.

For M1: suppose that $\M,s\vDash \nabla(\phi\vee\psi)\land\nabla(\neg\phi\vee\chi)$. Then $(\phi\vee\psi)^\M\notin N(s)$ and $(\neg\phi\vee\chi)^\M\notin N(s)$, that is, $\phi^\M\cup \psi^\M\notin N(s)$ and $(\neg\phi)^\M\cup \chi^\M\notin N(s)$. Since $\phi^\M\subseteq \phi^\M\cup \psi^\M$ and $N(s)$ is closed under supersets, we must have $\phi^\M\notin N(s)$. Similarly, we can show that $(\neg\phi)^\M\notin N(s)$. Therefore, $\M,s\vDash\nabla\phi$, as desired.
Similarly, we can show the validity of M2 and M3.
\weg{For M2: suppose that $\M,s\vDash\bullet(\phi\vee\psi)\land\bullet(\neg\phi\vee\chi)$. Then $(\phi\vee\psi)^\M\notin N(s)$ and $(\neg\phi\vee\chi)^\M\notin N(s)$, that is, $\phi^\M\cup\psi^\M\notin N(s)$ and $(\neg\phi)^\M\cup\chi^\M\notin N(s)$. Since $N(s)$ is closed under supersets, we obtain $\phi^\M\notin N(s)$ and $(\neg\phi)^\M\notin N(s)$. Therefore, $\M,s\vDash\nabla\phi$.

For M3: assume that $\M,s\vDash\bullet(\phi\vee\psi)\land\nabla(\neg\phi\vee\chi)$. Then $(\phi\vee\psi)^\M\notin N(s)$ and $(\neg\phi\vee\chi)\notin N(s)$, that is, $\phi^\M\cup\psi^\M\notin N(s)$ and $(\neg\phi)^\M\cup \chi^\M\notin N(s)$. By the property $(s)$, we infer that $\phi^\M\notin N(s)$ and $(\neg\phi)^\M\notin N(s)$. Therefore, $\M,s\vDash\nabla\phi$.}

For M4: assume that $\M,s\vDash\circ\phi\land\phi$. Then $\phi^\M\in N(s)$. Since $\phi^\M\subseteq \phi^\M\cup\psi^\M=(\phi\vee\psi)^\M$, by the property $(s)$, we have $(\phi\vee\psi)^\M\in N(s)$, and therefore $\M,s\vDash\Delta(\phi\vee\psi)\land\circ(\phi\vee\psi)$.
\end{proof}

\begin{definition}\label{def.cm-M}
Let $\Lambda$ be an extension of ${\bf M^{\nabla\bullet}}$. A triple $\M^\Lambda=\lr{S^\Lambda,N^\Lambda,V^\Lambda}$ is a {\em canonical neighborhood model} for $\Lambda$ if
\begin{itemize}
    \item $S^\Lambda=\{s\mid s\text{ is a maximal consistent set for }\Lambda\}$,
    \item $|\phi|\in N^\Lambda(s)$ iff $\Delta(\phi\vee\psi)\land\circ(\phi\vee\psi)\in s$ for all $\psi$,
    \item $V^\Lambda(p)=|p|=\{s\in S^\Lambda\mid p\in s\}$.
\end{itemize}
\end{definition}

We need to show that $N^\Lambda$ is well defined.
\begin{proposition}
Let $s\in S^\Lambda$ as defined in Def.~\ref{def.cm-M}. If $|\phi|=|\phi'|$, then $\Delta(\phi\vee\psi)\land\circ(\phi\vee\psi)\in s$ for all $\psi$ iff $\Delta(\phi'\vee\psi)\land\circ(\phi'\vee\psi)\in s$ for all $\psi$.
\end{proposition}

\begin{proof}
Suppose that $|\phi|=|\phi'|$, then $\vdash\phi\lra\phi'$, and thus $\vdash\phi\vee\psi\lra \phi'\vee\psi$. By $\text{RE}\nabla$, $\text{RE}\bullet$, $\text{Def.~}\Delta$ and $\text{Def.~}\circ$, we infer that $\vdash\Delta(\phi\vee\psi)\lra \Delta(\phi'\vee\psi)$ and $\vdash\circ(\phi\vee\psi)\lra \circ(\phi'\vee\psi)$, and hence $\vdash\Delta(\phi\vee\psi)\land\circ(\phi\vee\psi)\lra \Delta(\phi'\vee\psi)\land\circ(\phi'\vee\psi)$. Therefore, $\Delta(\phi\vee\psi)\land\circ(\phi\vee\psi)\in s$ for all $\psi$ iff $\Delta(\phi'\vee\psi)\land\circ(\phi'\vee\psi)\in s$ for all $\psi$.
\end{proof}

%$|\phi|\in N^c(s)$ iff $\Delta\phi\land\circ(\phi\vee\psi)\in s$ for all $\psi$

\begin{lemma}\label{lem.truthlem-m}
Let $\M^\Lambda=\lr{S^\Lambda,N^\Lambda,V^\Lambda}$ be an arbitrary canonical neighborhood model for any system $\Lambda$ extending ${\bf M^{\nabla\bullet}}$. Then for all $s\in S^\Lambda$, for all $\phi\in\mathcal{L}(\nabla,\bullet)$, we have
$$\M^\Lambda,s\vDash\phi\iff \phi\in s.$$
That is, $\phi^{\M^\Lambda}=|\phi|$.
\end{lemma}

\begin{proof}
By induction on $\phi$. The nontrivial cases are $\nabla\phi$ and $\bullet\phi$.

For case $\nabla\phi$:

Suppose that $\M^\Lambda,s\nvDash\nabla\phi$, to show that $\nabla\phi\notin s$. By supposition and induction hypothesis, $|\phi|\in N^\Lambda(s)$ or $S\backslash|\phi|\in N^\Lambda(s)$ (that is, $|\neg\phi|\in N^\Lambda(s)$). If $|\phi|\in N^\Lambda(s)$, then $\Delta(\phi\vee\psi)\in s$ for all $\psi$. By letting $\psi=\bot$, then $\Delta\phi\in s$, and thus $\nabla\phi\notin s$. If $|\neg\phi|\in N^\Lambda(s)$, with a similar argument we can show that $\Delta\neg\phi\in s$, that is, $\Delta\phi\in s$, and we also have $\nabla\phi\notin s$.

Conversely, assume that $\M^\Lambda,s\vDash\nabla\phi$, to prove that $\nabla\phi\in s$. By assumption and induction hypothesis, $|\phi|\notin N^\Lambda(s)$ and $S\backslash|\phi|\notin N^\Lambda(s)$, that is, $|\neg\phi|\notin N^\Lambda(s)$. Then $\Delta(\phi\vee\psi)\land\circ(\phi\vee\psi)\notin s$ for some $\psi$, and $\Delta(\neg\phi\vee\chi)\land\circ(\neg\phi\vee\chi)\notin  s$ for some $\chi$. We consider the following cases.
\begin{itemize}
    \item[-] $\Delta(\phi\vee\psi)\notin s$ and $\Delta(\neg\phi\vee\chi)\notin s$. That is, $\nabla(\phi\vee\psi)\in s$ and $\nabla(\neg\phi\vee\chi)\in s$. Then by axiom \text{M1}, we infer that $\nabla\phi\in s$.
    \item[-] $\Delta(\phi\vee\psi)\notin s$ and  $\circ(\neg\phi\vee\chi)\notin  s$. That is, $\nabla(\phi\vee\psi)\in s$ and $\bullet(\neg\phi\vee\chi)\in s$. By axiom $\text{M3}$, $\nabla\neg\phi\in s$, that is, $\nabla\phi\in s$.
    \item[-] $\circ(\phi\vee\psi)\notin s$ and $\Delta(\neg\phi\vee\chi)\notin s$. That is, $\bullet(\phi\vee\psi)\in s$ and $\nabla(\neg\phi\vee\chi)\in s$. Then by axiom $\text{M3}$, we derive that $\nabla\phi\in s$.
    \item[-] $\circ(\phi\vee\psi)\notin s$ and $\circ(\neg\phi\vee\chi)\notin s$. That is, $\bullet(\phi\vee\psi)\in s$ and $\bullet(\neg\phi\vee\chi)\in s$. By axiom $\text{M2}$, we obtain that $\nabla\phi\in s$.
\end{itemize}
Either case implies that $\nabla\phi\in s$, as desired.

For case $\bullet\phi$.

Suppose that $\bullet\phi\in s$, to show that $\M^\Lambda,s\vDash\bullet\phi$. By supposition and axiom $\text{E2}$, we obtain $\phi\in s$, which by induction hypothesis means that $\M^\Lambda,s\vDash\phi$. We have also $|\phi|\notin N^\Lambda(s)$: otherwise, by definition of $N^\Lambda$, we should have $\circ(\phi\vee\psi)\in s$ for all $\psi$, which then implies that $\circ\phi\in s$ (by letting $\psi=\bot$), a contradiction. Then by induction hypothesis, $\phi^{\M^\Lambda}\notin N^\Lambda(s)$. Therefore, $\M^\Lambda,s\vDash\bullet\phi$.

Conversely, assume that $\bullet\phi\notin s$ (that is, $\circ\phi\in s$), to prove that $\M^\Lambda,s\nvDash\bullet\phi$. For this, suppose that $\M^\Lambda,s\vDash\phi$, by induction hypothesis, we have $\phi\in s$, and then $\circ\phi\land\phi\in s$. By axiom $\text{M4}$, $\Delta(\phi\vee\psi)\land\circ(\phi\vee\psi)\in s$ for all $\psi$. By definition of $N^\Lambda$, we derive that $|\phi|\in N^\Lambda(s)$. Then by induction hypothesis again, we conclude that $\phi^{\M^\Lambda}\in N^\Lambda(s)$. Therefore, $\M^\Lambda,s\nvDash\bullet\phi$, as desired.
\end{proof}

Given an extension $\Lambda$ of ${\bf M^{\nabla\bullet}}$, the minimal canonical neighborhood model for $\Lambda$, denoted $\M^\Lambda_0=\lr{S^\Lambda,N^\Lambda_0,V^\Lambda}$, is defined such that $N^\Lambda_0(s)=\{|\phi|\mid \Delta(\phi\vee\psi)\land\circ(\phi\vee\psi)\in s\text{ for all }\psi\}$. Note that $\M^\Lambda_0$ is not necessarily supplemented. Therefore, we define a notion of supplementation, which comes from~\cite{Chellas1980}.
\begin{definition}
Let $\M=\lr{S,N,V}$ be a neighborhood model. The {\em supplementation} of $\M$, denoted $\M^+$, is a triple $\lr{S,N^+,V}$, in which for every $s\in S$, $N^+(s)$ is the superset closure of $N(s)$; namely, for each $s\in S$,
$$N^+(s)=\{X\subseteq S\mid Y\subseteq X\text{ for some }Y\in N(s)\}.$$
\end{definition}

One may easily show that $\M^+$ is supplemented, that is, $\M^+$ possesses $(s)$. Also, $N(s)\subseteq N^+(s)$. Moreover, the properties of being closed under intersections and containing the unit are closed under the supplementation.
\begin{proposition}\label{prop.properties-preserved}
Let $\M=\lr{S,N,V}$ be a neighborhood model and $\M^+$ be its supplementation. If $\M$ possesses $(i)$, then so does $\M^+$; if $\M$ possesses $(n)$, then so does $\M^+$.
\end{proposition}

In what follows, we will use $(\M^\Lambda_0)^+$ to denote the supplementation of $\M^\Lambda_0$, namely $(\M^\Lambda_0)^+=\lr{S^\Lambda,(N^\Lambda_0)^+,V^\Lambda}$, where $\Lambda$ extends ${\bf M^{\nabla\bullet}}$. By the definition of supplementation, $(\M^\Lambda_0)^+$ is an $(s)$-model. To show the completeness of ${\bf M^{\nabla\bullet}}$ over the class of $(s)$-frames, by Lemma~\ref{lem.truthlem-m}, it remains only to show that $(\M^\Lambda_0)^+$ is a canonical neighborhood model for $\Lambda$.

\begin{lemma}
Let $\Lambda$ extends ${\bf M^{\nabla\bullet}}$. For every $s\in S^\Lambda$, we have
$$|\phi|\in (N^\Lambda_0)^+(s)\iff \Delta(\phi\vee\psi)\land\circ(\phi\vee\psi)\in s\text{ for all }\psi.$$
\end{lemma}

\begin{proof}
Right-to-Left: Immediate by the definition of $N^\Lambda_0$ and the fact that $N^\Lambda_0(s)\subseteq (N^\Lambda_0)^+(s)$.

Left-to-Right: Suppose that $|\phi|\in (N^\Lambda_0)^+(s)$, to prove that $\Delta(\phi\vee\psi)\land\circ(\phi\vee\psi)\in s\text{ for all }\psi.$ By supposition, $X\subseteq |\phi|$ for some $X\in N^\Lambda_0(s)$. Then there must be a $\chi$ such that $X=|\chi|$, and thus $\Delta(\chi\vee\psi)\land\circ(\chi\vee\psi)\in s$ for all $\psi$, and hence $\Delta(\chi\vee\phi\vee\psi)\land\circ(\chi\vee\phi\vee\psi)\in s$. From $|\chi|\subseteq |\phi|$, it follows that $\vdash\chi\to\phi$, and then $\vdash\chi\vee\phi\vee\psi\lra \phi\vee\psi$, and thus $\vdash\Delta(\chi\vee\phi\vee\psi)\lra \Delta(\phi\vee\psi)$ and $\vdash \circ(\chi\vee\phi\vee\psi)\lra \circ(\phi\vee\psi)$ by $\text{RE}\nabla$, $\text{RE}\bullet$, $\text{Def.~}\Delta$ and $\text{Def.~}\circ$. Therefore, $\Delta(\phi\vee\psi)\land\circ(\phi\vee\psi)\in s$ for all $\psi$.
\end{proof}

Based on the previous analysis, we have the following.
\begin{theorem}\label{thm.m-comp}
${\bf M^{\nabla\bullet}}$ is sound and strongly complete with respect to the class of all $(s)$-frames.
\end{theorem}

We conclude this part with some results which will be used in Section~\ref{sec.updating}. The following result states that if one is ignorant of the fact that either $\phi$ holds or one is ignorant whether $\phi$ holds, then one is either ignorant of the fact that $\phi$ or ignorant whether $\phi$ holds.
%if it is accidental that either $\phi$ holds or $\phi$ is contingent, then $\phi$ is either accidental or contingent.
\begin{proposition}\label{prop.useful}
$\bullet(\phi\vee \nabla\phi)\to (\bullet \phi\vee\nabla \phi)$ is provable in ${\bf M^{\nabla\bullet}}$.\footnote{In fact, we can show a stronger result: $\bullet(\phi\vee\nabla\psi)\to\bullet\phi\vee\nabla\psi$ is provable in ${\bf M^{\nabla\bullet}}$. But we do not need such a strong result below.}
\end{proposition}

\begin{proof}
By Thm.~\ref{thm.m-comp}, it suffices to show the formula is valid over the class of $(s)$-frames. Let $\M=\lr{S,N,V}$ be an $(s)$-model and $s\in S$.

Suppose, for reductio, that $\M,s\vDash\bullet(\phi\vee\nabla \phi)$ and $\M,s\nvDash \bullet \phi\vee\nabla \phi$. From the former, it follows that $\M,s\vDash  \phi\vee\nabla \phi$ and $(\phi\vee\nabla  \phi)^\M\notin N(s)$; from the latter, it follows that $\M,s\nvDash \bullet \phi$ and $\M,s\nvDash \nabla \phi$. This implies that $\M,s\vDash \phi$, which plus $\M,s\nvDash\bullet \phi$ gives us $\phi^\M\in N(s)$. Since $\phi^\M\subseteq(\phi\vee\nabla \phi)^\M$, by $(s)$, we conclude that $(\phi\vee\nabla \phi)^\M\in N(s)$: a contradiction, as desired.
\end{proof}

%The following result says that if it is accidental that either the non-accident or non-contingency of $\phi$ implies $\phi$, then it is accidental that $\phi$.
The following result says that if one is ignorant of the fact that either non-ignorance of $\phi$ or non-ignorance whether $\phi$ holds implies that $\phi$, then one is ignorant of the fact that $\phi$.
\begin{proposition}\label{prop.useful-2}
$\bullet(\circ \phi\vee\Delta \phi\to \phi)\to \bullet \phi$ is provable in ${\bf M^{\nabla\bullet}}$.
\end{proposition}

\begin{proof}
By Thm.~\ref{thm.m-comp}, it remains only to prove that the formula is valid over the class of $(s)$-frames. Let $\M=\lr{S,N,V}$ be an $(s)$-model and $s\in S$.

Assume, for reductio, that $\M,s\vDash \bullet(\circ\phi\vee\Delta\phi\to \phi)$ and $\M,s\nvDash \bullet\phi$. The former implies $\M,s\vDash \circ\phi\vee\Delta\phi\to \phi$ and $(\circ\phi\vee\Delta\phi\to \phi)^\M\notin N(s)$; the latter entails that $\M,s\vDash \circ\phi$. Then $\M,s\vDash\phi$, and thus $\phi^\M\in N(s)$. One may easily verify that $\phi^\M\subseteq (\circ\phi\vee\Delta\phi\to \phi)^\M$. Then by $(s)$, we conclude that $(\circ\phi\vee\Delta\phi\to \phi)^\M\in N(s)$: a contradiction.
\end{proof}

\subsubsection{Regular logic}

Define ${\bf R^{\nabla\bullet}}:={\bf M^{\nabla\bullet}}+\text{R1}+\text{R2}$, where
\[\begin{array}{ll}
    \text{R1} & \Delta\phi\land\Delta\psi\to\Delta(\phi\land\psi) \\
    \text{R2} & \circ\phi\land\circ\psi\to\circ(\phi\land\psi)\\
\end{array}\]

\begin{proposition}\label{prop.soundness-r}
${\bf R^{\nabla\bullet}}$ is sound with respect to the class of quasi-filters.
\end{proposition}

\begin{proof}
By soundness of ${\bf M^{\nabla\bullet}}$, it remains to prove the validity of R1 and R2. The validity of R1 has been shown in~\cite[Prop.~3(iv)]{fan2019family}, and the validity of R2 has been shown in~\cite[Thm.~5.2]{Fan:2020neighborhood}.
\end{proof}

\begin{proposition}\label{prop.extends-r}
Let $\Lambda$ extends ${\bf R^{\nabla\bullet}}$. Then the minimal canonical model $\M^\Lambda_0$ has the property $(i)$. As a corollary, its supplementation is a quasi-filter.
\end{proposition}

\begin{proof}
Suppose that $X,Y\in N^\Lambda_0(s)$, to show that $X\cap Y\in N^\Lambda_0(s)$. By supposition, there exist $\phi$ and $\chi$ such that $X=|\phi|$ and $Y=|\chi|$, and then $\Delta(\phi\vee\psi)\land\circ(\phi\vee\psi)\in s$ for all $\psi$, and $\Delta(\chi\vee\psi)\land\circ(\chi\vee\psi)\in s$ for all $\psi$. By axioms R1 and R2, we can obtain that $\Delta((\phi\land\chi)\vee\psi)\land\circ((\phi\land\chi)\vee\psi)\in s$ for all $\psi$, which implies that $|\phi\land\chi|\in N^\Lambda_0(s)$, that is, $X\cap Y\in N^\Lambda_0(s)$.
\end{proof}

\begin{theorem}\label{thm.r-comp}
${\bf R^{\nabla\bullet}}$ is sound and strongly complete with respect to the class of quasi-filters.
\end{theorem}

\subsubsection{${\bf K^{\nabla\bullet}}$}

Define ${\bf K^{\nabla\bullet}}:={\bf R^{\nabla\bullet}}+\circ\top$.

Again, like the case of ${\bf EN^{\nabla\bullet}}$(Sec.~\ref{sec.en}), $\Delta\top$ is derivable from
$\circ\top$ and Prop.~\ref{prop.circimpliesdelta}. This hints us that the inference rule $\text{R1}$ in~\cite[Def.~12]{Fan:2019} is actually dispensable. (Fact 13 therein is derivable from axiom A1 and axiom A6. Then by R2, we have $\vdash\phi$ implies $\vdash \circ\phi\land\phi$, and then $\vdash \Delta\phi$. Thus we derive R1 there.)

\begin{theorem}\label{thm.comp-k-nabla-bullet}
${\bf K^{\nabla\bullet}}$ is sound and strongly complete with respect to the class of filters.
\end{theorem}

\begin{proof}
For soundness, by Prop.~\ref{prop.soundness-r}, it suffices to show the validity of $\circ\top$ over the class of filters. This follows immediately from Prop.~\ref{prop.n-definable}.

For completeness, define $\M^\Lambda_0$ as before w.r.t. ${\bf K^{\nabla\bullet}}$. By Prop.~\ref{prop.extends-r} and Prop.~\ref{prop.properties-preserved}, it remains only to show that $N^\Lambda_0(s)$ possesses $(n)$. By $\circ\top$ and the derivable formula $\Delta\top$, we have $\vdash \circ\top$ and $\vdash \Delta\top$. Then by $\vdash (\top\vee\psi)\lra \top$,  $\text{RE}\nabla$, $\text{RE}\bullet$, $\text{Def.~}\Delta$ and $\text{Def.~}\circ$, we infer that for all $s\in S^\Lambda$, $\Delta(\top\vee\psi)\land\circ(\top\vee\psi)\in s$ for all $\psi$, and thus $|\top|\in N^\Lambda(s)$, that is, $S\in N^\Lambda_0(s)$, as desired.
\end{proof}

Inspired by the definition of $N^\Lambda$, one may define the canonical relation for the extensions of ${\bf K^{\nabla\bullet}}$ as follows:

%$sR^Nt$ iff for all $\phi$, if $\Delta(\phi\vee\psi)\land\circ(\phi\vee\chi)\in s$ for all $\psi$ and $\chi$, then $\phi\in t$.
$sR^Nt$ iff for all $\phi$, if $\Delta(\phi\vee\psi)\land\circ(\phi\vee\psi)\in s$ for all $\psi$, then $\phi\in t$.

Recall that the original definition of canonical relation given in~\cite[Def.~18]{Fan:2019} is as follows:

$sR^Kt$ iff there exists $\delta$ such that $(a)$ $\bullet\delta\in s$, and $(b)$ for all $\phi$, if $\Delta\phi\land\circ(\neg\delta\to\phi)\in s$, then $\phi\in t$.

One may ask what the relationship between $R^N$ and $R^K$ is. As we shall see, they are equal to each other. Before this, we need some preparation.

\begin{proposition}\label{prop.preparation}
$\vdash \bullet\delta\land\Delta\phi\land\circ(\neg\delta\to\phi)\to\Delta(\phi\vee\psi)\land\circ(\phi\vee\chi)$
\end{proposition}

\begin{proof}
%One may show that this formula is valid over the class of filters. Then use the completeness of ${\bf K^{\nabla\bullet}}$ (Thm.~\ref{thm.comp-k-nabla-bullet}).
By Thm.~\ref{thm.comp-k-nabla-bullet}, it remains only to show that this formula is valid over the class of filters.

Let $\M=\lr{S,N,V}$ be a filter and $s\in S$. Suppose that $\M,s\vDash \bullet\delta\land\Delta\phi\land\circ(\neg\delta\to\phi)$, to show $\M,s\vDash \Delta(\phi\vee\psi)\land\circ(\phi\vee\chi)$. By $\M,s\vDash \bullet\delta$, we have $\M,s\vDash\delta$ and $\delta^\M\notin N(s)$. By $\M,s\vDash\Delta\phi$, we infer that $\phi^\M\in N(s)$ or $S\backslash\phi^\M\in N(s)$. Since $\M,s\vDash\delta$, we derive that $\M,s\vdash\neg\delta\to\phi$. Then by $\M,s\vDash\circ(\neg\delta\to\phi)$, we get $(\neg\delta\to\phi)^\M\in N(s)$, that is, $\delta^\M\cup \phi^\M\in N(s)$. If $S\backslash\phi^\M\in N(s)$, then as $N(s)$ has the property $(i)$, $(\delta^\M\cup \phi^\M)\cap(S\backslash\phi^\M)\in N(s)$, viz. $\delta^\M\cap (S\backslash\phi^\M)\in N(s)$. Since $N(s)$ possesses the property $(s)$ and $\delta^\M\cap(S\backslash\phi^\M)\subseteq \delta^\M$, it follows that $\delta^\M\in N(s)$: a contradiction. This entails that $S\backslash\phi^\M\notin N(s)$, and thus $\phi^\M\in N(s)$. Note that $\phi^\M\subseteq \phi^\M\cup\psi^\M=(\phi\vee\psi)^\M$ and $\phi^\M\subseteq \phi^\M\cup\chi^\M=(\phi\vee\chi)^\M$. Using $(s)$ again, we conclude that $(\phi\vee\psi)^\M\in N(s)$  and $(\phi\vee\chi)^\M\in N(s)$, and therefore $\M,s\vDash\Delta(\phi\vee\psi)\land \circ(\phi\vee\chi)$, as desired.
\end{proof}

\begin{proposition}
Let $\Lambda$ be an extension of ${\bf K^{\nabla\bullet}}$. Then for all $s,t\in S^\Lambda$, $sR^Nt$ iff $sR^Kt$.%\footnote{This is claimed without proofs in~\cite[fn.2]{Fan:2021}.}
\end{proposition}

\begin{proof}
Suppose that $sR^Nt$, to show that $sR^Kt$. By supposition, for all $\phi$, if $\Delta(\phi\vee\psi)\land\circ(\phi\vee\psi)\in s$ for all $\psi$, then $\phi\in t$. Letting $\phi=\bot$, we can infer that $\Delta\psi\land\circ\psi\notin s$ for some $\psi$. If $\Delta\psi\notin s$, that is, $\nabla\psi\in s$, then by axiom $\text{E3}$, we derive that $\bullet\psi\in s$ or $\bullet\neg\psi\in s$. If $\circ\psi\notin s$, then we have $\bullet\psi\in s$. Either case implies that $\bullet\delta\in s$ for some $\delta$. Now suppose for any $\phi'$ that $\Delta\phi'\land\circ(\neg\delta\to\phi')\in s$. By Prop.~\ref{prop.preparation}, we infer that $\Delta(\phi'\vee\chi)\land\circ(\phi'\vee\chi)\in s$ for all $\chi$. Then by supposition again, we conclude that $\phi'\in t$. Therefore, $sR^Kt$.

Conversely, assume that $sR^Kt$, then there exists $\delta$ such that $(a)$ $\bullet\delta\in s$, and $(b)$ for all $\phi$, if $\Delta\phi\land\circ(\neg\delta\to\phi)\in s$, then $\phi\in t$. It remains to prove that $sR^Nt$. For this, suppose for any $\phi$ that $\Delta(\phi\vee\psi)\land\circ(\phi\vee\psi)\in s$ for all $\psi$. By letting $\psi=\bot$, we obtain that $\Delta\phi\in s$; by letting $\psi=\delta$, we infer that $\circ(\neg\delta\to\phi)\in s$. Thus $\Delta\phi\land \circ(\neg \delta\to\phi)\in s$. Then by $(b)$, we conclude that $\phi\in t$, and therefore $sR^Nt$, as desired. 
\end{proof}

\weg{\begin{proof}
Suppose that $sR^Nt$, to show that $sR^Kt$. By supposition, for all $\phi$, if $\Delta(\phi\vee\psi)\land\circ(\phi\vee\chi)\in s$ for all $\psi$ and $\chi$, then $\phi\in t$. Letting $\phi=\bot$, we can infer that $\Delta\psi\land\circ\chi\notin s$ for some $\psi$ and some $\chi$. If $\Delta\psi\notin s$, that is, $\nabla\psi\in s$, then by axiom $\text{E3}$, we derive that $\bullet\psi\in s$ or $\bullet\neg\psi\in s$. If $\circ\chi\notin s$, then we have $\bullet\chi\in s$. Either case implies that $\bullet\delta\in s$ for some $\delta$. Now suppose for any $\phi'$ that $\Delta\phi'\land\circ(\neg\delta\to\phi')\in s$. By Prop.~\ref{prop.preparation}, we infer that $\Delta(\phi'\vee\psi)\land\circ(\phi'\vee\chi)\in s$ for all $\psi$ and all $\chi$. Then by supposition again, we conclude that $\phi'\in t$.

Conversely, assume that $sR^Kt$, then there exists $\delta$ such that $(a)$ $\bullet\delta\in s$, and $(b)$ for all $\phi$, if $\Delta\phi\land\circ(\neg\delta\to\phi)\in s$, then $\phi\in t$. It remains to prove that $sR^Nt$. For this, suppose for any $\phi$ that $\Delta(\phi\vee\psi)\land\circ(\phi\vee\chi)\in s$ for all $\psi$ and $\chi$. By letting $\psi=\bot$ and $\chi=\delta$, we obtain that $\Delta\phi\land\circ(\neg\delta\to\phi)\in s$. Then by $(b)$, we conclude that $\phi\in t$, as desired. 
\end{proof}}

%\section{Adding public announcements}
%\section{Two ways of updating neighborhood models}
\section{Updating neighborhood models}\label{sec.updating}

%The intersection semantics and the subset semantics

In this section, we extend the previous results to the dynamic case of public announcements. Syntactically, we add the constructor $[\phi]\phi$ into the previous languages $\mathcal{L}(\nabla)$, $\mathcal{L}(\bullet)$ and $\mathcal{L}(\nabla,\bullet)$, and denote the obtained extensions by $\mathcal{L}(\nabla,[\cdot])$, $\mathcal{L}(\bullet,[\cdot])$, $\mathcal{L}(\nabla,\bullet,[\cdot])$, respectively. $[\psi]\phi$ is read ``after every truthfully public announcement of $\psi$, $\phi$ holds''. Also, as usual, $\lr{\psi}\phi$ abbreviates $\neg[\psi]\neg\phi$. Semantically, we adopt the intersection semantics in the literature (e.g.~\cite{Zvesper:2010,ma2013update,MaSano:2015}). In detail, given a monotone neighborhood model $\M=\lr{S,N,V}$ and a state $s\in S$,
%The intersection semantics
\[
\begin{array}{lll}
    \M,s\vDash[\psi]\phi & \iff & \M,s\vDash \psi \text{ implies }\M^{\cap\psi},s\vDash\phi  \\
\end{array}
\]
where $\M^{\cap\psi}$ is the intersection submodel $\M^{\cap\psi^\M}$, and the notion of intersection submodels is defined as below.

\begin{definition}\cite[Def.~3]{ma2013update}
Let $\M=\lr{S,N,V}$ be a monotone model, and $X$ a nonempty subset of $S$. Define {\em the intersection model} $\M^{\cap X}=\lr{X,N^{\cap X},V^X}$ induced from $X$ in the following.
\begin{itemize}
    \item for every $s\in X$, $N^{\cap X}=\{Y\mid Y=P\cap X\text{ for some }P\in N(s)\}$,
    \item $V^X(p)=V(p)\cap X$.
\end{itemize}
\end{definition}

\begin{proposition}\cite[Prop.~2]{ma2013update}
The neighborhood property $(s)$ is preserved under taking the intersection submodel. That is, if $\M$ is a monotone neighborhood model with the domain $S$, then for any nonempty subset $X$ of $S$, the intersection submodel $\M^{\cap X}$ is also monotone.
\end{proposition}

The following lists the reduction axioms of $\mathcal{L}(\nabla,\bullet,[\cdot])$ and its sublanguages $\mathcal{L}(\nabla,[\cdot])$ and $\mathcal{L}(\bullet,[\cdot])$ under the intersection semantics.
%${\bf M^{\nabla\bullet}_\texttt{Int}}$
\[
\begin{array}{ll}
    \text{AP} & [\psi]p\lra (\psi\to p) \\
    \text{AN} & [\psi]\neg\phi\lra (\psi\to\neg[\psi]\phi)\\
    \text{AC} & [\psi](\phi\land\chi)\lra ([\psi]\phi\land[\psi]\chi)\\
    \text{AA} & [\psi][\chi]\phi\lra [\psi\land[\psi]\chi]\phi\\
    \text{A}\nabla & [\psi]\nabla\phi\lra (\psi\to\nabla[\psi]\phi\land\nabla[\psi]\neg\phi)\\
    \text{A}\bullet&{[\psi]\bullet}\phi\lra (\psi\to \bullet[\psi]\phi)\\
\end{array}
\]

The following reduction axioms are derivable from the above reduction axioms.
\[
\begin{array}{ll}
    \text{A}\Delta & [\psi]\Delta\phi\lra (\psi\to\Delta[\psi]\phi\vee\Delta[\psi]\neg\phi) \\
    \text{A}\circ & [\psi]{\circ\phi}\lra (\psi\to\circ[\psi]\phi)\\
\end{array}
\]

\begin{theorem}
Let $\Lambda$ be a system of $\mathcal{L}(\nabla)$ (resp. $\mathcal{L}(\bullet)$, $\mathcal{L}(\nabla,\bullet)$). If $\Lambda$ is sound and strongly complete with respect to the class of monotone neighborhood frames, then so is $\Lambda$ plus $\text{AP}$, $\text{AN}$, $\text{AC}$, $\text{AA}$ and $\text{A}\nabla$ (resp. plus $\text{AP}$, $\text{AN}$, $\text{AC}$, $\text{AA}$ and $\text{A}\bullet$, plus $\text{AP}$, $\text{AN}$, $\text{AC}$, $\text{AA}$, $\text{A}\nabla$ and $\text{A}\bullet$) under intersection semantics.
\end{theorem}

\begin{proof}
The validity of axioms AP, AN, AC, AA can be found in~\cite[Thm.~1]{ma2013update},~\cite[Thm.~2, Thm.~3]{MaSano:2015} and~\cite[Prop.~3.1]{Zvesper:2010}. The validity of $\text{A}\bullet$ has been shown in~\cite{Fan:2020neighborhood}, where the axiom is named $\texttt{A}{\bullet\texttt{Int}}$. The validity of $\text{A}\nabla$ is shown as follows.
Let $\M=\lr{S,N,V}$ be an $(s)$-model and $s\in S$. 

To begin with, suppose that $\M,s\vDash [\psi]\nabla\phi$ and $\M,s\vDash\psi$, to show that $\M,s\vDash\nabla[\psi]\phi\land\nabla[\psi]\neg\phi$. By supposition, we have $\M^{\cap \psi},s\vDash\nabla\phi$, which implies $\phi^{\M^{\cap \psi}}\notin N^{\cap \psi}(s)$ and $\psi^\M\backslash\phi^{\M^{\cap\psi}}\notin N^{\cap\psi}(s)$. 

We claim that $\M,s\vDash\nabla[\psi]\phi$, that is, $([\psi]\phi)^\M\notin N(s)$ and $S\backslash([\psi]\phi)^\M\notin N(s)$. If $([\psi]\phi)^\M\in N(s)$, then $([\psi]\phi)^\M\cap \psi^\M\in N^{\cap \psi}(s)$. As $([\psi]\phi)^\M\cap \psi^\M\subseteq \phi^{\M^{\cap\psi}}$, by $(s)$, we have $\phi^{\M^{\cap\psi}}\in N^{\cap\psi}(s)$: a contradiction. If $S\backslash([\psi]\phi)^\M\in N(s)$, then $(S\backslash([\psi]\phi)^\M)\cap \psi^\M\in N^{\cap\psi}(s)$. Note that $(S\backslash([\psi]\phi)^\M)\cap \psi^\M\subseteq \psi^\M\backslash\phi^{\M^{\cap\psi}}$: for any $x\in (S\backslash([\psi]\phi)^\M)\cap \psi^\M$, $x\notin([\psi]\phi)^\M$, thus $x\in \psi^\M$ and $x\notin \phi^{\M^{\cap\psi}}$, and hence $x\in \psi^\M\backslash\phi^{\M^{\cap\psi}}$. By $(s)$ again, $\psi^\M\backslash\phi^{\M^{\cap\psi}}\in N^{\cap\psi}(s)$: a contradiction again.

We also claim that  $\M,s\vDash\nabla[\psi]\neg\phi$, that is, $([\psi]\neg\phi)^\M\notin N(s)$ and $S\backslash([\psi]\neg\phi)^\M\notin N(s)$. If $([\psi]\neg\phi)^\M\in N(s)$, then $([\psi]\neg\phi)^\M\cap \psi^\M\in N^{\cap\psi}(s)$. As $([\psi]\neg\phi)^\M\cap \psi^\M\subseteq \psi^\M\backslash\phi^{\M^{\cap\psi}}$, we infer by $(s)$ that $\psi^\M\backslash\phi^{\M^{\cap\psi}}\in N^{\cap\psi}(s)$: a contradiction. If $S\backslash([\psi]\neg\phi)^\M\in N(s)$, then $(S\backslash([\psi]\neg\phi)^\M)\cap\psi^\M\in N^{\cap\psi}(s)$. Since $(S\backslash([\psi]\neg\phi)^\M)\cap\psi^\M\subseteq \phi^{\M^{\cap\psi}}$, by $(s)$ again, we derive that $\phi^{\M^{\cap\psi}}\in N^{\cap\psi}(s)$: a contradiction again.

Conversely, assume that $\M,s\vDash\psi\to\nabla[\psi]\phi\land\nabla[\psi]\neg\phi$, to prove that $\M,s\vDash[\psi]\nabla\phi$. For this, we suppose that $\M,s\vDash\psi$, it remains only to show that $\M^{\cap\psi},s\vDash\nabla\phi$, that is, $\phi^{\M^{\cap\psi}}\notin N^{\cap\psi}(s)$ and $\psi^\M\backslash\phi^{\M^{\cap\psi}}\notin N^{\cap\psi}(s)$. By assumption and supposition, we obtain $\M,s\vDash\nabla[\psi]\phi\land\nabla[\psi]\neg\phi$. This follows that $([\psi]\phi)^\M\notin N(s)$ and $S\backslash([\psi]\phi)^\M\notin N(s)$, and $([\psi]\neg\phi)^\M\notin N(s)$ and $S\backslash([\psi]\neg\phi)^\M\notin N(s)$.

We claim that $\phi^{\M^{\cap\psi}}\notin N^{\cap\psi}(s)$. Otherwise, that is, $\phi^{\M^{\cap\psi}}\in N^{\cap\psi}(s)$, we have $\phi^{\M^{\cap\psi}}=P\cap \psi^\M$ for some $P\in N(s)$. This implies that $P\subseteq ([\psi]\phi)^\M$: for any $x\in P$, we would have $x\in([\psi]\phi)^\M$, since if $x\in \psi^\M$, then $x\in P\cap\psi^\M=\phi^{\M^{\cap\psi}}$. By $(s)$, $([\psi]\phi)^\M\in N(s)$: a contradiction.

We also claim that $\psi^\M\backslash\phi^{\M^{\cap\psi}}\notin N^{\cap\psi}(s)$. Otherwise, that is, $\psi^\M\backslash\phi^{\M^{\cap\psi}}\in N^{\cap\psi}(s)$, we infer that $\psi^\M\backslash\phi^{\M^{\cap\psi}}=P\cap \psi^\M$ for some $P\in N(s)$. It then follows that $P\subseteq ([\psi]\neg\phi)^\M$: for any $x\in P$, we have $x\in ([\psi]\neg\phi)^\M$, since if $x\in \psi^\M$, then $x=P\cap\psi^\M=\psi^\M\backslash \phi^{\M^{\cap\psi}}$, and so $x\in (\neg\phi)^{\M^{\cap\psi}}$. By $(s)$ again, $([\psi]\neg\phi)^\M\in N(s)$: a contradiction, as desired.
\end{proof}

For the sake of reference, we use ${\bf M^{\nabla\bullet[\cdot]}}$ to denote the extension of ${\bf M^{\nabla\bullet}}$ with all the above reduction axioms. By dropping $\text{A}\bullet$ from ${\bf M^{\nabla\bullet[\cdot]}}$, we obtain the system ${\bf M^{\nabla[\cdot]}}$; by dropping $\text{A}\nabla$ from ${\bf M^{\nabla\bullet[\cdot]}}$, we obtain the system ${\bf M^{\bullet[\cdot]}}$.

In what follows, we will focus on some successful formulas in our languages. A formula is said to be {\em successful}, if it still holds after being announced; in symbols, $\vDash[\phi]\phi$. Recall that $\neg{\bullet p}$ is shown to be successful under the relational semantics in~\cite[Prop.~39]{Fan:2019} and under the intersection semantics in~\cite[Prop.~6.5]{Fan:2020neighborhood}. We will follow this line of research and say much more. As we shall show, any combination of $p$, $\neg p$, ${\neg\bullet}p$, and $\neg\nabla p$ via conjunction (or, via disjunction) is successful under the intersection semantics.\footnote{It is shown in~\cite[Prop.~38]{Fan:2019} that under Kripke semantics, $\bullet p$ is self-refuting and $\neg{\bullet p}$ is successful.} %This is the main result of this section.

%Some results about successful in $\mathcal{L}(\bullet,[\cdot])$ under relational semantics have been obtained in. 

%\cite{Fan:2020neighborhood} shows that the formula $[\neg\bullet p]\neg\bullet p$ are provable in ${\bf M^{\bullet[\cdot]}}$, which means that the negation of Moore sentences are successful. Axiom A$\nabla$ extends the result of ${\bf K^{\nabla[\cdot]}}$ in~\cite{Fanetal:2015} to ${\bf M^{\nabla[\cdot]}}$.\footnote{${\bf K^{\nabla[\cdot]}}$ is called $\mathbb{CLA}$ in~\cite[Def.~7.3]{Fanetal:2015}. There, a reduction axiom $[\phi]\Delta\psi\lra (\phi\to(\Delta[\phi]\psi\vee\Delta[\phi]\neg\psi))$ is proposed, from which and the definition of $\nabla$ we can obtain the reduction axiom A$\nabla$.}

\weg{\begin{theorem}
${\bf M^{\nabla\bullet[\cdot]}}$ is sound and strongly complete with respect to the class of $(s)$-frames.
\end{theorem}}

%It turns out that any conjunct of $p\land {\neg\bullet} p\land \neg\nabla p$ is successful under the intersection semantics.

To begin with, we show that, provably, any combination of $p$, $\neg p$, ${\neg\bullet}p$, and $\neg\nabla p$ via {\em conjunction} is successful under the intersection semantics.

\begin{proposition}
$p$ is successful under the intersection semantics. That is, $[p]p$ is provable in ${\bf M^{\bullet[\cdot]}}$. 
\end{proposition}

\begin{proof}
Straightforward by $\text{AP}$.
\end{proof}

\begin{proposition}\label{prop.notpissuccessful}
$\neg p$ is successful under the intersection semantics. That is, $[\neg p]\neg p$ is provable in ${\bf M^{\bullet[\cdot]}}$. 
\end{proposition}

\begin{proof}
Straightforward by $\text{AN}$ and $\text{AP}$.
\end{proof}

\begin{proposition}\label{prop.notbulletpissuccessful}
${\neg\bullet}p$ is successful under the intersection semantics. That is, $[{\neg\bullet} p]{\neg\bullet} p$ is provable in ${\bf M^{\bullet[\cdot]}}$.
\end{proposition}

\begin{proof}
Refer to~\cite[Prop.~6.5]{Fan:2020neighborhood}.
\end{proof}

\begin{proposition}
$\neg\nabla p$ is successful under the intersection semantics. That is, $[\neg\nabla p]\neg\nabla p$ is provable in ${\bf M^{\nabla[\cdot]}}$.
\end{proposition}

\begin{proof}
We have the following proof sequence in ${\bf M^{\nabla[\cdot]}}$.
\[
\begin{array}{lll}
   & [\neg\nabla p]\neg\nabla p\\ 
   \lra &
    (\neg\nabla p\to \neg[\neg\nabla p]\nabla p) & \text{AN}\\
    \lra & (\neg\nabla p\to \neg (\neg\nabla p\to\nabla[\neg\nabla p]p\land \nabla[\neg\nabla p]\neg p) & \text{A}\nabla\\
      \lra & (\neg\nabla p\to \neg(\nabla (\neg\nabla p\to p)\land \nabla (\neg\nabla p\to \neg (\neg\nabla p\to p)))) & \text{AP},\text{AN} \\
      \lra &(\nabla(\neg\nabla p\to p)\land \nabla (\neg \nabla p\to \neg p)\to \nabla p) & \text{TAUT},\text{RE}\nabla\\
      \lra & (\nabla(p\vee \nabla p)\land \nabla (\neg p\vee\nabla p)\to \nabla p) & \text{TAUT},\text{RE}\nabla\\
     \lra & \top & \text{M1}\\
\end{array}
\]
Therefore, $[\neg\nabla p]\neg\nabla p$ is provable in ${\bf M^{\nabla[\cdot]}}$.
\end{proof}

Intuitively, $[\neg\nabla p]\neg\nabla p$ means that ``after being told that one is not ignorant whether $p$, one is still not ignorant whether $p$.'' In other words, one's non-ignorance about a fact cannot be altered by being announced.

\begin{proposition}\label{prop.notpandpissuccessful}
$p\land\neg  p$ is successful under the intersection semantics.
\end{proposition}

\begin{proof}
Note that $p\land\neg  p$ is equivalent to $\bot$, and $\bot$ is successful.
\end{proof}

\begin{proposition}
$p\land{\neg\bullet} p$ is successful under the intersection semantics. That is, $[p\land{\neg\bullet} p](p\land{\neg\bullet} p)$ is provable in ${\bf M^{\bullet[\cdot]}}$.
\end{proposition}

\begin{proof}
We have the following proof sequence in ${\bf M^{\bullet[\cdot]}}$.
\[
\begin{array}{lll}
  & [p\land{\neg\bullet} p](p\land{\neg\bullet} p) & \\
  \lra &  ([p\land{\neg\bullet} p]p\land [p\land{\neg\bullet} p]{\neg\bullet} p) & \text{AC} \\
    \lra & (p\land{\neg\bullet} p\to p)\land (p\land{\neg\bullet} p\to \neg[p\land{\neg\bullet} p]{\bullet} p) & \text{AP},\text{AN}\\
    \lra  &  (p\land{\neg\bullet} p\to \neg(p\land{\neg\bullet} p\to {\bullet[p\land{\neg\bullet} p]} p) & \text{A}\bullet\\
    \lra & (p\land{\neg\bullet} p\to {\neg\bullet}[p\land{\neg\bullet} p] p) & \text{TAUT}\\
    \lra & (p\land{\neg\bullet}p\to {\neg\bullet}(p\land{\neg\bullet}p\to p)) & \text{AP}\\
    \lra & (p\land{\neg\bullet} p\to {\neg \bullet} \top) & \text{TAUT},\text{RE}\bullet\\
    \lra & (p\land \circ p\to\circ\top) & \text{Def.~}\circ\\
    \lra & \top & \text{Prop.~}\ref{prop.impliescirctop}\\
\end{array}
\]
Therefore, $[p\land{\neg\bullet} p](p\land{\neg\bullet} p)$ is provable in ${\bf M^{\bullet[\cdot]}}$.
\end{proof}

\begin{proposition}\label{prop.pandnablapissuccessful}
$p\land\neg\nabla p$ is successful under the intersection semantics. That is, $[p\land\neg\nabla p](p\land\neg\nabla p)$ is provable in ${\bf M^{\nabla[\cdot]}}$.
\end{proposition}

\begin{proof}
We have the following proof sequence in ${\bf M^{\nabla[\cdot]}}$.
\[
\begin{array}{lll}
    & [p\land\neg\nabla p](p\land\neg\nabla p) & \\
    \lra & ([p\land\neg\nabla p]p\land [p\land\neg\nabla p]\neg\nabla p) & \text{AC} \\
     \lra & (p\land\neg\nabla p\to p)\land (p\land\neg\nabla p\to \neg[p\land\neg\nabla p]\nabla p) & \text{AP},\text{AN}\\
     \lra & (p\land\neg\nabla p\to \neg(p\land\neg\nabla p\to\nabla [p\land\neg\nabla p]p\land \nabla [p\land\neg\nabla p]\neg p)) & \text{A}\nabla\\
     \lra & (p\land \neg\nabla p\to \neg(\nabla[p\land\neg\nabla p]p\land\nabla [p\land\neg\nabla p]\neg p)) & \text{TAUT}\\
     \lra & (p\land\neg\nabla p\to \neg (\nabla \top\land \nabla [p\land\neg\nabla p]\neg p))& \text{AP},\text{RE}\nabla\\
     \lra & (p\land\neg\nabla p\to \neg\nabla \top\vee \neg\nabla (p\land\neg\nabla p\to \neg p))& \text{AN},\text{AP},\text{RE}\nabla\\
     \lra & (p\land\Delta p\to \Delta \top\vee \Delta (p\land\Delta p\to \neg p))& \text{Def.~}\Delta\\
\end{array}
\]
By Prop.~\ref{prop.impliesdeltatop}, $\Delta p\to \Delta \top$ is provable in ${\bf M^{\nabla}}$, so is the last formula in the above proof sequence, and thus $[p\land\neg\nabla p](p\land\neg\nabla p)$ is provable in ${\bf M^{\nabla[\cdot]}}$.
\end{proof}

\begin{proposition}\label{prop.notpandnotbulletpissuccessful}
$\neg p\land{\neg\bullet}p$ is successful under the intersection semantics.
\end{proposition}

\begin{proof}
By $\text{E2}$, $\neg p\land{\neg\bullet}p$ is equivalent to $\neg p$. And we have already known from Prop.~\ref{prop.notpissuccessful} that $\neg p$ is successful.
\end{proof}

\begin{proposition}\label{prop.notpandnotnablapissuccessful}
$\neg p\land\neg\nabla p$ is successful under the intersection semantics. That is, $[\neg p\land\neg\nabla p](\neg p\land\neg\nabla p)$ is provable in ${\bf M^{\nabla[\cdot]}}$.
\end{proposition}

\begin{proof}
We have the following proof sequence in ${\bf M^{\nabla[\cdot]}}$.
\[
\begin{array}{lll}
     & [\neg p\land\neg\nabla p](\neg p\land\neg\nabla p) & \\
    \lra & [\neg p\land\neg\nabla p]\neg p\land [\neg p\land\neg\nabla p]\neg\nabla p & \text{AC}\\
    \lra & (\neg p\land\neg\nabla p\to \neg p)\land (\neg p\land\neg\nabla p\to\neg[\neg p\land\neg\nabla p]\nabla p) & \text{AN},\text{AP}\\
    \lra &  (\neg p\land\neg\nabla p\to\neg[\neg p\land\neg\nabla p]\nabla p) & \text{TAUT}\\
    \lra & (\neg p\land\neg\nabla p\to\neg(\nabla[\neg p\land\neg\nabla p]p\land \nabla[\neg p\land\neg\nabla p]\neg p)&\text{A}\nabla\\
    \lra & (\neg p\land\Delta p\to \Delta (\neg p\land\Delta p\to p)\vee\Delta \top) & \text{AP},\text{AN},\text{RE}\nabla\\
    \lra & \top & \text{Prop.~}\ref{prop.impliesdeltatop}\\
\end{array}
\]
Therefore, $[\neg p\land\neg\nabla p](\neg p\land\neg\nabla p)$ is provable in ${\bf M^{\nabla[\cdot]}}$.
\end{proof}

We have seen that both ${\neg\bullet} p$ and $\neg\nabla p$ are successful. One natural question would be whether their conjunction, viz. ${\neg\bullet} p\land \neg\nabla p$, is successful. Note that this does not obviously hold, since for instance, both $p$ and $\neg K p$ are successful, whereas $p\land\neg Kp$ is not, see e.g.~\cite[Example~4.34]{hvdetal.del:2007}.
\begin{proposition}
${\neg\bullet} p\land \neg\nabla p$ is successful under the intersection semantics. That is, $[{\neg\bullet} p\land \neg\nabla p]({\neg\bullet} p\land \neg\nabla p)$ is provable in ${\bf M^{\nabla\bullet[\cdot]}}$.
\end{proposition}

\begin{proof}
By $\text{AC}$, $[{\neg\bullet} p\land \neg\nabla p]({\neg\bullet} p\land \neg\nabla p) \lra ([{\neg\bullet} p\land \neg\nabla p]{\neg\bullet} p\land [{\neg\bullet} p\land \neg\nabla p]\neg\nabla p)$. We show that both $[{\neg\bullet} p\land \neg\nabla p]{\neg\bullet} p$ and $[{\neg\bullet} p\land \neg\nabla p]\neg\nabla p$ are provable in ${\bf M^{\nabla\bullet[\cdot]}}$.

We have the following proof sequence in ${\bf M^{\nabla\bullet[\cdot]}}$.
\[
\begin{array}{lll}
  & [{\neg\bullet} p\land \neg\nabla p]{\neg\bullet} p  & \\
  \lra & ({\neg\bullet} p\land \neg\nabla p\to \neg [{\neg\bullet} p\land \neg\nabla p]{\bullet p}) & \text{AN} \\
     \lra & ({\neg\bullet} p\land \neg\nabla p\to\neg ({\neg\bullet} p\land \neg\nabla p\to \bullet[{\neg\bullet} p\land \neg\nabla p]p)& \text{A}\bullet\\
    \lra & ({\neg\bullet} p\land \neg\nabla p\to{\neg\bullet}[{\neg\bullet} p\land \neg\nabla p]p)&\text{TAUT}\\
    \lra & ({\neg\bullet} p\land \neg\nabla p\to{\neg\bullet}({\neg\bullet} p\land \neg\nabla p\to p)) & \text{AP},\text{RE}\bullet\\
    \lra & (\bullet({\neg\bullet} p\land \neg\nabla p\to p)\to (\bullet p\vee\nabla p)) & \text{TAUT}\\
    \lra & (\bullet({\bullet} p\lor \nabla p\vee p)\to (\bullet p\vee\nabla p)) &\text{TAUT},\text{RE}\bullet\\
     \lra & (\bullet(p\vee\nabla p)\to (\bullet p\vee\nabla p)) & \text{E2},\text{RE}\bullet\\
     \lra & \top & \text{Prop.~}\ref{prop.useful}\\
\end{array}
\]
\[
\begin{array}{lll}
  & [{\neg\bullet} p\land \neg\nabla p]\neg\nabla p & \\
  \lra &   ({\neg\bullet} p\land \neg\nabla p\to \neg [{\neg\bullet} p\land \neg\nabla p]{\nabla p}) & \text{AN} \\
     \lra &  ({\neg\bullet} p\land \neg\nabla p\to \neg({\neg\bullet} p\land \neg\nabla p\to \nabla[{\neg\bullet} p\land \neg\nabla p]p\land &\\ &\nabla[{\neg\bullet} p\land \neg\nabla p]\neg p)) & \text{A}\nabla \\
     \lra & ({\neg\bullet} p\land \neg\nabla p\to \neg(\nabla[{\neg\bullet} p\land \neg\nabla p]p\land\nabla[{\neg\bullet} p\land \neg\nabla p]\neg p))& \text{TAUT}\\ 
     \lra & ({\neg\bullet} p\land \neg\nabla p\to \neg(\nabla({\neg\bullet} p\land \neg\nabla p\to p)\land\nabla({\neg\bullet} p\land \neg\nabla p\to\neg p)) & \text{AP},\text{AN}\\
     \lra & (\nabla (p\vee\bullet p\vee\nabla p)\land \nabla (\neg p\vee \bullet p\lor \nabla p)\to (\nabla p\vee \bullet p)) & \text{TAUT},\text{RE}\bullet\\
    \lra & \top & \text{M1}\\
\end{array}
\]

Thus both $[{\neg\bullet} p\land \neg\nabla p]{\neg\bullet} p$ and $[{\neg\bullet} p\land \neg\nabla p]\neg\nabla p$ are provable in ${\bf M^{\nabla\bullet[\cdot]}}$. Therefore, $[{\neg\bullet} p\land \neg\nabla p]({\neg\bullet} p\land \neg\nabla p)$ is provable in ${\bf M^{\nabla\bullet[\cdot]}}$.
\end{proof}

Intuitively, $[{\neg\bullet} p\land \neg\nabla p]({\neg\bullet} p\land \neg\nabla p)$ says that after being told that one is neither ignorant whether nor 
ignorant of $p$, one is still neither ignorant whether nor ignorant of $p$. In short, one's non-ignorance whether and non-ignorance of a fact cannot be altered by being announced.

The following two propositions can be shown as in Prop.~\ref{prop.notpandpissuccessful}.

\begin{proposition}
$p\land\neg  p\land {\neg\bullet} p$ is successful under the intersection semantics.
\end{proposition}

\begin{proposition}
$p\land\neg  p\land \neg\nabla p$ is successful under the intersection semantics.
\end{proposition}

\begin{proposition}\label{prop.threeconjunct-successful}
$p\land{\neg\bullet} p\land\neg\nabla p$ is successful under the intersection semantics. That is, $[p\land{\neg\bullet} p\land\neg\nabla p](p\land{\neg\bullet} p\land\neg\nabla p)$ is provable in ${\bf M^{\nabla\bullet[\cdot]}}$.
\end{proposition}

\begin{proof}
By $\text{AC}$, $[p\land{\neg\bullet} p\land\neg\nabla p](p\land{\neg\bullet} p\land\neg\nabla p)\lra ([p\land{\neg\bullet} p\land\neg\nabla p]p\land [p\land{\neg\bullet} p\land\neg\nabla p]{\neg\bullet} p\land [p\land{\neg\bullet} p\land\neg\nabla p]\neg\nabla p)$. One may easily verify that $[p\land{\neg\bullet} p\land\neg\nabla p]p$ is provable in ${\bf M^{\nabla\bullet[\cdot]}}$. It remains only to show that both $[p\land{\neg\bullet} p\land\neg\nabla p]{\neg\bullet} p$ and $[p\land{\neg\bullet} p\land\neg\nabla p]\neg\nabla p$ are provable in the system in question.

We have the following proof sequence in ${\bf M^{\nabla\bullet[\cdot]}}$.
\[
\begin{array}{lll}
   & [p\land{\neg\bullet} p\land\neg\nabla p]{\neg\bullet} p & \\
   \lra &  (p\land{\neg\bullet} p\land\neg\nabla p\to \neg[p\land{\neg\bullet} p\land\neg\nabla p]{\bullet p}) & \text{AN} \\
    \lra & (p\land{\neg\bullet} p\land\neg\nabla p\to{\neg\bullet}[p\land{\neg\bullet} p\land\neg\nabla p]p) & \text{A}\bullet\\
    \lra & (p\land{\neg\bullet} p\land\neg\nabla p\to{\neg\bullet}\top) & \text{AP},\text{RE}\bullet\\
    \lra & (p\land \circ p\land \Delta p\to \circ\top) & \text{Def.~}\circ,\text{Def.~}\Delta\\
\end{array}
\]
By Prop.~\ref{prop.impliescirctop}, $p\land\circ p\to \circ\top$ is provable in ${\bf M^\bullet}$, so is the last formula in the above proof sequence, and thus $[p\land{\neg\bullet} p\land\neg\nabla p]{\neg\bullet} p$ is provable in ${\bf M^{\nabla{\bullet[\cdot]}}}$.

Also, we have the following proof sequence in ${\bf M^{\nabla{\bullet[\cdot]}}}$.
\[
\begin{array}{lll}
   & [p\land{\neg\bullet} p\land\neg\nabla p]\neg\nabla p &\\ \lra & (p\land{\neg\bullet} p\land\neg\nabla p\to \neg[p\land{\neg\bullet} p\land\neg\nabla p]\nabla p)  & \text{AN} \\
     \lra & (p\land{\neg\bullet} p\land\neg\nabla p\to \neg(\nabla[p\land{\neg\bullet} p\land\neg\nabla p] p\land \nabla[p\land{\neg\bullet} p\land\neg\nabla p]\neg p)) & \text{A}\nabla\\
     \lra & (p\land{\circ} p\land\Delta p\to \Delta [p\land{\neg\bullet} p\land\neg\nabla p] p\lor\Delta [p\land{\neg\bullet} p\land\neg\nabla p]\neg p) & \text{Def.~}\circ,\text{Def.~}\Delta\\
     \lra & (p\land{\circ} p\land\Delta p\to \Delta \top\lor\Delta [p\land{\neg\bullet} p\land\neg\nabla p]\neg p) & \text{AP},\text{RE}\nabla\\
\end{array}
\]
By Prop.~\ref{prop.impliesdeltatop}, $\Delta p\to \Delta \top$ is provable in ${\bf M^{\nabla}}$, thus the last formula in the above proof sequence is provable in ${\bf M^{\nabla\bullet}}$. Therefore, $[p\land{\neg\bullet} p\land\neg\nabla p]\neg\nabla p$ is provable in ${\bf M^{\nabla\bullet[\cdot]}}$.

According to the previous analysis, $[p\land{\neg\bullet} p\land\neg\nabla p](p\land{\neg\bullet} p\land\neg\nabla p)$ is provable in ${\bf M^{\nabla\bullet[\cdot]}}$.
\end{proof}

\begin{proposition}
$\neg p\land {\neg\bullet} p\land\neg\nabla p$ is successful under the intersection semantics. That is, $[\neg p\land{\neg\bullet} p\land\neg\nabla p](\neg p\land{\neg\bullet} p\land\neg\nabla p)$ is provable in ${\bf M^{\nabla\bullet[\cdot]}}$.
\end{proposition}

\begin{proof}
By the proof of Prop.~\ref{prop.notpandnotbulletpissuccessful}, $\neg p\land{\neg\bullet} p\land\neg\nabla p$ is equivalent to $\neg p\land \neg \nabla p$. And we have already shown in Prop.~\ref{prop.notpandnotnablapissuccessful} that $\neg p\land \neg \nabla p$ is successful under the intersection semantics.
\end{proof}

\weg{\begin{proof}
By $\text{AC}$, $[\neg p\land{\neg\bullet} p\land\neg\nabla p](\neg p\land{\neg\bullet} p\land\neg\nabla p)\lra ([\neg p\land{\neg\bullet} p\land\neg\nabla p]\neg p\land [\neg p\land{\neg\bullet} p\land\neg\nabla p]{\neg\bullet} p\land [\neg p\land{\neg\bullet} p\land\neg\nabla p]\neg\nabla p)$. One may easily verify that $[\neg p\land{\neg\bullet} p\land\neg\nabla p]\neg p$ is provable in the system in question. It remains only to show that $[\neg p\land{\neg\bullet} p\land\neg\nabla p]{\neg\bullet} p$ and $[\neg p\land{\neg\bullet} p\land\neg\nabla p]\neg\nabla p$ are provable in that system.

We have the following proof sequence.
\[
\begin{array}{lll}
     & [\neg p\land{\neg\bullet} p\land\neg\nabla p]{\neg\bullet} p & \\
    \lra & (\neg p\land{\neg\bullet} p\land\neg\nabla p\to \neg[\neg p\land{\neg\bullet} p\land\neg\nabla p]{\bullet p}) & \text{AN} \\
    \lra & (\neg p\land{\neg\bullet} p\land\neg\nabla p\to {\neg\bullet}[\neg p\land{\neg\bullet} p\land\neg\nabla p]p & \text{A}\bullet\\
\end{array}
\]
\end{proof}}

\begin{proposition}
$p\land\neg  p\land {\neg\bullet} p\land \neg\nabla p$ is successful under the intersection semantics.
\end{proposition}

\begin{proof}
The proof is similar to that of Prop.~\ref{prop.notpandpissuccessful}.
\end{proof}

Now we demonstrate that any combination of $p$, $\neg p$, ${\neg\bullet}p$, and $\neg\nabla p$ via {\em disjunction} is successful under the intersection semantics. First, one may show that $[\psi](\phi\vee\chi)\lra ([\psi]\phi\vee[\psi]\chi)$ is provable from the above reduction axioms. For the sake of reference, we denote it $\text{AD}$.

\begin{proposition}\label{prop.pornotpissuccessful}
$p\vee\neg p$ is successful under the intersection semantics.
\end{proposition}

\begin{proof}
Note that $p\vee \neg p$ is equivalent to $\top$, and $\top$ is successful.
\end{proof}

\begin{proposition}\label{prop.pornotbulletpissuccessful}
$p\vee{\neg\bullet }p$ is successful under the intersection semantics. That is, $[p\vee{\neg\bullet }p](p\vee{\neg\bullet }p)$ is provable in ${\bf M^{\bullet[\cdot]}}$.
\end{proposition}

\begin{proof}
Just note that $p\vee{\neg\bullet} p$ is equivalent to $\bullet p\to p$, which by $\text{E2}$ is equivalent to $\top$. And $\top$ is successful.
\end{proof}

\weg{\begin{proof}
We have the following proof sequence in ${\bf M^{\bullet[\cdot]}}$.
\[
\begin{array}{lll}
   & [p\vee{\neg\bullet }p](p\vee{\neg\bullet }p) & \\
   \lra & ([p\vee{\neg\bullet }p]p\vee [p\vee{\neg\bullet }p]{\neg\bullet }p)  & \text{AD} \\
     \lra & (p\vee{\neg\bullet }p\to p)\vee (p\vee{\neg\bullet }p\to \neg[p\vee{\neg\bullet }p]{\bullet p}) & \text{AP},\text{AN}\\
     \lra & (p\vee{\neg\bullet }p\to p)\vee(p\vee{\neg\bullet }p\to {\neg\bullet}[p\vee{\neg\bullet }p]p) & \text{A}\bullet\\
     \lra & (p\vee{\neg\bullet }p\to p)\vee(p\vee{\neg\bullet }p\to {\neg\bullet}(p\vee{\neg\bullet }p\to p)) & \text{AP}\\
     \lra & (p\vee{\neg\bullet }p\to p\vee {\neg\bullet}(p\vee{\neg\bullet }p\to p)) &\text{TAUT}\\
     \lra & ({\neg\bullet }p\to p\vee {\neg\bullet}({\neg\bullet }p\to p)) & \text{TAUT},\text{RE}\bullet\\
     \lra & (\bullet p\vee p\vee {\neg\bullet}({\neg\bullet }p\to p)) & \text{TAUT}\\
     \lra & (p\vee{\neg\bullet}({\circ }p\to p)) & \text{E2},\text{Def.~}\circ\\
     \lra & (\bullet(\circ p\to p)\to p) & \text{TAUT}\\
     \lra & \top & \text{Prop.~}\ref{prop.impliesphi}\\
\end{array}
\]
Therefore, $[p\vee{\neg\bullet }p](p\vee{\neg\bullet }p)$ is provable in ${\bf M^{\bullet[\cdot]}}$.
\end{proof}}

\begin{proposition}
$p\vee\neg\nabla p$ is successful under the intersection semantics. That is, $[p\vee\neg\nabla p](p\vee\neg\nabla p)$ is provable in ${\bf M^{\nabla[\cdot]}}$.
\end{proposition}

\begin{proof}
We have the following proof sequence in ${\bf M^{\nabla[\cdot]}}$.
\[
\begin{array}{lll}
     & [p\vee\neg\nabla p](p\vee\neg\nabla p) &\\
     \lra & ([p\vee\neg\nabla p]p\vee [p\vee\neg\nabla p]\neg\nabla p) &  \text{AD} \\
     \lra & (p\vee\neg\nabla p\to p)\vee (p\vee\neg\nabla p\to \neg[p\vee\neg\nabla p]\nabla p) & \text{AP},\text{AN}\\
     \lra & (p\vee\neg\nabla p\to p)\vee (p\vee\neg\nabla p\to\neg(\nabla[p\vee\neg\nabla p]p\land \nabla[p\vee\neg\nabla p]\neg p))& \text{A}\nabla\\
     \lra &  (p\vee\neg\nabla p\to p)\vee (p\vee\neg\nabla p\to\neg(\nabla(p\vee\neg\nabla p\to p)\land \nabla(p\vee\neg\nabla p\to\neg p))) & \text{AP},\text{AN}\\
     \lra & (p\vee\neg\nabla p\to p\vee\neg(\nabla(p\vee\neg\nabla p\to p)\land \nabla(p\vee\neg\nabla p\to\neg p))) & \text{TAUT}\\
     \lra & (\neg p\land \nabla(p\vee\neg(p\vee\neg\nabla p) )\land \nabla(\neg p\vee\neg(p\vee\neg\nabla p))\to \neg p\land\nabla p) & \text{TAUT},\text{RE}\nabla\\
     \lra & \top & \text{M1}\\
\end{array}
\]
Therefore, $[p\vee\neg\nabla p](p\vee\neg\nabla p)$ is provable in ${\bf M^{\nabla[\cdot]}}$.
\end{proof}

\begin{proposition}
$\neg p\vee{\neg\bullet}p$ is successful under the intersection semantics.
\end{proposition}

\begin{proof}
By $\text{E2}$, $\neg p\vee{\neg\bullet}p$ is equivalent to ${\neg\bullet}p$, and Prop.~\ref{prop.notbulletpissuccessful} has shown that ${\neg\bullet}p$ is successful under the intersection semantics. 
\end{proof}

%F1 and F2 mean that, respectively, ``Moore sentences are self-refuting'' and ``''. As a matter of fact, both formulas are provable in ${\bf M^{\bullet[\cdot]}}$.

\begin{proposition}
$\neg p\vee\neg \nabla p$ is successful under the intersection semantics. That is, $[\neg p\vee\neg \nabla p](\neg p\vee\neg \nabla p)$ is provable in ${\bf M^{\nabla[\cdot]}}$.
\end{proposition}

\begin{proof}
We have the following proof sequence in ${\bf M^{\nabla[\cdot]}}$.
\[
\begin{array}{lll}
     &  [\neg p\vee\neg \nabla p](\neg p\vee\neg \nabla p) & \\
    \lra & [\neg p\vee\neg \nabla p]\neg p\vee [\neg p\vee\neg \nabla p]\neg\nabla p & \text{AD}\\
    \lra & (\neg p\vee\neg \nabla p\to \neg p)\vee (\neg p\vee\neg\nabla p\to\neg[\neg p\vee\neg \nabla p]\nabla p) & \text{AN},\text{AP}\\
    \lra & (\neg p\vee\neg \nabla p\to \neg p)\vee (\neg p\vee\neg\nabla p\to\neg(\nabla[\neg p\vee\neg\nabla p]p\land \nabla[\neg p\vee\neg\nabla p]\neg p))& \text{A}\nabla \\
    \lra & (\neg p\vee\neg \nabla p\to \neg p)\vee (\neg p\vee\neg\nabla p\to\neg(\nabla(\neg p\vee\neg\nabla p\to p)\land \nabla(\neg p\vee\neg\nabla p\to\neg p))) &\text{AP},\text{AN}\\
    \lra & (\neg p\vee\neg \nabla p\to \neg p\vee \neg(\nabla(\neg p\vee\neg\nabla p\to p)\land \nabla(\neg p\vee\neg\nabla p\to\neg p))) & \text{TAUT}\\
    \lra & p\land \nabla (p\vee \neg(\neg p\vee\neg\nabla p))\land \nabla(\neg p\vee \neg(\neg p\vee\neg\nabla p))\to p\land\nabla p& \text{TAUT}\\
    \lra & \top & \text{M1}\\
\end{array}
\]
Therefore, $[\neg p\vee\neg \nabla p](\neg p\vee\neg \nabla p)$ is provable in ${\bf M^{\nabla[\cdot]}}$.
\end{proof}

\begin{proposition}\label{prop.notbulletpornotnablapissuccessful}
${\neg\bullet}p\vee\neg\nabla p$ is successful under the intersection semantics. That is, $[{\neg\bullet}p\vee\neg\nabla p]({\neg\bullet}p\vee\neg\nabla p)$ is provable in ${\bf M^{\nabla\bullet[\cdot]}}$.
\end{proposition}

\begin{proof}
%Since $[{\neg\bullet}p\vee\neg\nabla p]({\neg\bullet}p\vee\neg\nabla p)\lra ([{\neg\bullet}p\vee\neg\nabla p]{\neg\bullet}p\vee [{\neg\bullet}p\vee\neg\nabla p]{\neg\nabla}p)$, it suffices to show the provability of $[{\neg\bullet}p\vee\neg\nabla p]{\neg\bullet}p\vee [{\neg\bullet}p\vee\neg\nabla p]{\neg\nabla}p$.
We have the following proof sequence in ${\bf M^{\nabla\bullet[\cdot]}}$.
\[
\begin{array}{lll}
   & [{\neg\bullet}p\vee\neg\nabla p]({\neg\bullet}p\vee\neg\nabla p) &\\ 
   \lra & ([{\neg\bullet}p\vee\neg\nabla p]{\neg\bullet}p\vee [{\neg\bullet}p\vee\neg\nabla p]{\neg\nabla}p) & \text{AD} \\
    \lra & ({\neg\bullet}p\vee\neg\nabla p\to \neg[{\neg\bullet}p\vee\neg\nabla p]{\bullet}p)\vee &\\
     & ({\neg\bullet}p\vee\neg\nabla p\to \neg[{\neg\bullet}p\vee\neg\nabla p]{\nabla}p) & \text{AN}\\
     \lra & ({\neg\bullet}p\vee\neg\nabla p\to {\neg\bullet}[{\neg\bullet}p\vee\neg\nabla p]p)\vee &\\
     & ({\neg\bullet}p\vee\neg\nabla p\to \neg({\nabla}[{\neg\bullet}p\vee\neg\nabla p]p\land & \\
     & {\nabla}[{\neg\bullet}p\vee\neg\nabla p]\neg p)) & \text{A}\bullet,\text{A}\nabla\\
     \lra & ({\neg\bullet}p\vee\neg\nabla p\to {\neg\bullet}({\neg\bullet}p\vee\neg\nabla p\to p))\vee \\
    & ({\neg\bullet}p\vee\neg\nabla p\to \neg({\nabla}({\neg\bullet}p\vee\neg\nabla p\to p)\land & \\
    & {\nabla}({\neg\bullet}p\vee\neg\nabla p\to \neg p))) & \text{AP},\text{AN}\\
    \lra & ({\neg\bullet}p\vee\neg\nabla p\to{\neg\bullet}({\neg\bullet}p\vee\neg\nabla p\to p)\vee & \\ 
    & \neg({\nabla}({\neg\bullet}p\vee\neg\nabla p\to p)\land {\nabla}({\neg\bullet}p\vee\neg\nabla p\to \neg p))) &\text{TAUT}\\
     \lra & {\bullet}({\neg\bullet}p\vee\neg\nabla p\to p)\land {\nabla}({\neg\bullet}p\vee\neg\nabla p\to p)\land & \\  
     & {\nabla}({\neg\bullet}p\vee\neg\nabla p\to \neg p)\to \bullet p\land \nabla p & \text{TAUT}\\
    \lra & \bullet(\circ p\vee\Delta p\to p)\land \nabla(\circ p\vee\Delta p\to p)\land & \\ &\nabla (\circ p\vee\Delta p\to\neg p)\to \bullet p\land \nabla p & \text{Def.~}\circ,\text{Def.~}\Delta \\
\end{array}
\]

By Prop.~\ref{prop.useful-2}, $\bullet(\circ p\vee\Delta p\to p)\to \bullet p$ is provable in ${\bf M^{\nabla\bullet}}$; by axiom $\text{M1}$ and $\text{RE}\nabla$, we can show the provability of $\nabla(\circ p\vee\Delta p\to p)\land\nabla (\circ p\vee\Delta p\to\neg p)\to\nabla p$ in ${\bf M^{\nabla\bullet}}$. Therefore, $[{\neg\bullet}p\vee\neg\nabla p]({\neg\bullet}p\vee\neg\nabla p)$ is provable in ${\bf M^{\nabla\bullet[\cdot]}}$.
\end{proof}

Intuitively, $[{\neg\bullet}p\vee\neg\nabla p]({\neg\bullet}p\vee\neg\nabla p)$ says that one's either non-ignorance of or non-ignorance whether a fact cannot be altered by being announced: after being told that one is either not ignorant of or not ignorant whether $p$, one is still either not ignorant of or not ignorant whether $p$.

Next two propositions are shown as in Prop.~\ref{prop.pornotpissuccessful}.

\begin{proposition}
$p\vee\neg p\vee{\neg\bullet }p$ is successful under the intersection semantics.%That is, $[p\vee{\neg\bullet }p\vee\neg\nabla p](p\vee{\neg\bullet }p\vee\neg\nabla p)$ is provable in ${\bf M^{\nabla\bullet[\cdot]}}$.
\end{proposition}

\begin{proposition}
$p\vee\neg p\vee{\neg\nabla }p$ is successful under the intersection semantics.%That is, $[p\vee{\neg\bullet }p\vee\neg\nabla p](p\vee{\neg\bullet }p\vee\neg\nabla p)$ is provable in ${\bf M^{\nabla\bullet[\cdot]}}$.
\end{proposition}

\begin{proposition}
$p\vee{\neg\bullet }p\vee\neg\nabla p$ is successful under the intersection semantics. That is, $[p\vee{\neg\bullet }p\vee\neg\nabla p](p\vee{\neg\bullet }p\vee\neg\nabla p)$ is provable in ${\bf M^{\nabla\bullet[\cdot]}}$.
\end{proposition}

\begin{proof}
By the proof of Prop.~\ref{prop.pornotbulletpissuccessful}, $p\vee{\neg\bullet }p$ is equivalent to $\top$, so is $p\vee{\neg\bullet }p\vee\neg\nabla p$.
And $\top$ is successful.
\end{proof}

\begin{proposition}
$\neg p\vee{\neg\bullet}p\vee\neg\nabla p$ is successful under the intersection semantics.
\end{proposition}

\begin{proof}
By $\text{E2}$, $\neg p\vee{\neg\bullet }p\vee\neg\nabla p$ is equivalent to ${\neg\bullet }p\vee\neg\nabla p$, and we have shown in Prop.~\ref{prop.notbulletpornotnablapissuccessful} that ${\neg\bullet }p\vee\neg\nabla p$ is successful under the intersection semantics.
\end{proof}

\begin{proposition}
$p\vee\neg p\vee{\neg\bullet}p\vee\neg\nabla p$ is successful under the intersection semantics.
\end{proposition}

\begin{proof}
The proof is similar to that of Prop.~\ref{prop.pornotpissuccessful}.
\end{proof}

%Question: Indicate under which conditions $\bullet p$ and $\nabla p$ and their Boolean combinations are successful and self-refuting.

%Q: indicate the fragment of successful formulas.

\section{Conclusion and Future work}\label{sec.conclusion}

In this paper, we investigated the bimodal logic of Fitchean ignorance and first-order ignorance $\mathcal{L}(\nabla,\bullet)$ under the neighborhood semantics. We compared the relative expressivity between $\mathcal{L}(\nabla,\bullet)$ and the logic of (first-order) ignorance $\mathcal{L}(\nabla)$ and the logic of Fitchean ignorance $\mathcal{L}(\bullet)$, and between $\mathcal{L}(\nabla,\bullet)$ and standard epistemic logic $\mathcal{L}(\Diamond)$. It turns out that over the class of models possessing $(c)$ or $(t)$, all of these logics are equally expressive, whereas over the class of models possessing either of other eight neighborhood properties, $\mathcal{L}(\nabla,\bullet)$ is more expressive than both $\mathcal{L}(\nabla)$ and $\mathcal{L}(\bullet)$, and over the class of models possessing either of eight neighborhood properties except for $(d)$,
$\mathcal{L}(\nabla,\bullet)$ is less expressive than $\mathcal{L}(\Diamond)$. We explored the frame definability of the bimodal logic, which turns out that all ten frame properties except for $(n)$ are undefinable in $\mathcal{L}(\nabla,\bullet)$. We axiomatized the bimodal logic over various classes of neighborhood frames, which among other results includes the classical logic, the monotone logic, and the regular logic. Last but not least, by updating the neighborhood models via the intersection semantics, we found suitable reduction axioms and thus reduced the public announcement operators to the bimodal logic. This gives us good applications to successful formulas, since as we have shown, any combination of $p$, $\neg p$, ${\neg\bullet}p$ and $\neg\nabla p$ via conjunction (or, via disjunction) is successful under the intersection semantics. We also partly answers open questions raised in~\cite{Fan:2021,Fan:2020neighborhood}.

For future work, we hope to know whether $\mathcal{L}(\nabla,\bullet)$ is less expressive than $\mathcal{L}(\Diamond)$
over the class of $(d)$-models. We conjecture the answer is positive, but the model
constructions seems hard, where the desired models both needs at least three points.
Moreover, as we have seen, the proofs of the expressivity and frame definability results involve nontrivial (if not highly nontrivial) constructions of neighborhood models
and frames, we thus also hope to find the bisimulation notion for $\mathcal{L}(\nabla,\bullet)$ under the
neighborhood semantics.

\bibliographystyle{plain}
\bibliography{biblio2019,biblio2016}

\end{document}